\documentclass[review,hidelinks,onefignum,onetabnum]{siamart251216}

\usepackage{lipsum}
\usepackage{amsfonts}
\usepackage{graphicx}
\usepackage{epstopdf}
\usepackage{algorithmic}
\ifpdf
  \DeclareGraphicsExtensions{.eps,.pdf,.png,.jpg}
\else
  \DeclareGraphicsExtensions{.eps}
\fi

\newsiamremark{remark}{Remark}
\newsiamremark{hypothesis}{Hypothesis}
\crefname{hypothesis}{Hypothesis}{Hypotheses}
\newsiamthm{claim}{Claim}
\newsiamremark{fact}{Fact}
\crefname{fact}{Fact}{Facts}

\title{Structure-preserving dynamical low-rank approximation for parametric elastic guided waves\thanks{Submitted to the editors DATE.
\funding{This work was supported by the French National Research Agency (ANR) under the ``Jeunes Chercheurs et Jeunes Chercheuses" program (Grant ANR-25-CE51-1125-01).}}}

\author{Dimitri Goutaudier\thanks{PIMM Laboratory, CNRS, Art et Métiers, CNAM, 151 boulevard de l'Hôpital, 75013 Paris
  (\email{dimitri.goutaudier@ensam.eu})}}

\usepackage{amsopn}

\ifpdf
\hypersetup{
  pdftitle={DLRA elastic guided waves},
  pdfauthor={D. Goutaudier}
}
\fi

\begin{document}

\maketitle

% REQUIRED
\begin{abstract}
Elastic guided waves are widely used in Structural Health Monitoring (SHM). In many-query settings, the computational cost of high-fidelity simulations motivates the use of projection-based reduced order modeling (ROM). However, the transport-dominated and dispersive nature of guided waves challenges static linear subspaces. In addition, preserving the Hamiltonian structure of the equations for energy conservation necessitates dedicated projection techniques. While the Dynamical Low Rank Approximation (DLRA) has proven effective for other wave equations, its application to elastic guided waves in SHM has remained unexplored. In this work, we introduce a structure-preserving parametric ROM framework that leverages the DLRA in an off-line/on-line strategy. During the off-line stage, a time-dependent symplectic reduced basis is constructed from training simulations. For a simplified class of parameter dependencies, we derive a closed-form solution of the nonlinear basis evolution equation. This analytical result yields a closed-form, energy-preserving reduced propagator during wave propagation, eliminating on-line time integration after the loading phase. We validate our approach on a 2D elasticity problem featuring dispersive guided waves interacting with a damage. The results demonstrate high compression ratios (rank $\sim 10-30$), low full field reconstruction errors ($\sim 10^{-3}-10^{-2}$), speedups of two to three orders of magnitude, and excellent long-time energy conservation. 
\end{abstract}

% REQUIRED

\section{Introduction}

The goal of projection-based reduced order modeling (ROM) is to accelerate simulations \cite{Benner_SIAM_2015, hesthaven2016certified} for multi-query or real time applications, such as uncertainty quantification \cite{chen2017reduced} or parametric studies \cite{Goutaudier_GTP_2024}. This is achieved by projecting  high-dimensional equations, typically arising from the discretization of time-dependent partial differential equations, onto a low-dimensional subspace. The efficiency of a ROM strategy critically depends on the choice of the reduced subspace and the projector, which must be tailored to the mathematical structure of the underlying physical system. Indeed, physical systems exhibit a range of dynamical behaviors, such as dissipative or conservative, linear or nonlinear, transport-dominated, stochastic, etc., each posing specific challenges. \\

In this study, we focus on Hamiltonian wave equation problems, involving energy conservation and moving features with high gradients. Such systems are central in various Structural Health Monitoring (SHM) applications, where elastic guided waves are used to detect, localize and characterize damages at early stages (e.g., corrosion in metallic parts \cite{nicard2025situ}, delamination in composite materials \cite{fendzi2016general}). However, the interplay between energy preservation, moving wave fronts and dispersion (frequency-dependent wave speeds), renders conventional ROM techniques ineffective \cite{peherstorfer2020model}. For instance, the well-established POD-Galerkin method fails for these problems at two fundamental levels. First, a static linear subspace cannot efficiently track translating features over an arbitrary time window. Several methods have been proposed to address this difficulty, such as high-order manifolds \cite{barnett2022quadratic}, registration methods \cite{taddei2021space}, the shifted Proper Orthogonal Decomposition \cite{reiss2018shifted}, the time-incremental Proper Generalized Decomposition \cite{goutaudier2021proper,goutaudier2022exploring}, or dynamic reduced bases \cite{hesthaven2022rank, pagliantini2021dynamical,musharbash2020symplectic}. Second, the orthogonal projection breaks the Hamiltonian structure, leading to energy drift and loss of physical fidelity. Structure-preserving model order reduction techniques address this challenge with symplectic Galerkin projection frameworks, involving a symplectic basis and an (oblique) symplectic projector. Such a reduced basis can be obtained by approximating the solution of the so-called Proper Symplectic Decomposition (PSD) problem \cite{peng2016symplectic}, which mirrors the traditional POD problem. In particular, the complex SVD method provides the solution of the PSD problem restricted to the set of ortho-symplectic bases. Successful applications of symplectic model order reduction have been reported for mechanical engineering applications \cite{buchfink2019symplectic, peng2022data}. However, a PSD basis remains static, which limits its applicability for the transport-dominated dynamics concerning this work. \\

The Dynamical Low Rank Approximation (DLRA) offers a promising framework for Hamiltonian wave equations, as it dynamically adapts the reduced basis to track moving features \cite{koch2007dynamical}. This method seeks a low rank approximation of a large time-dependent matrix satisfying a matrix differential equation. For instance, such situation arises in this study when simultaneously solving the elastodynamic equations at multiple query parameters. The DLRA proceeds by projecting the high-dimensional flow equations onto the tangent space of matrices of fixed rank. This yields differential equations for the factors of an SVD-like decomposition of the time-dependent low-rank approximation. Over the past decades, efficient time marching schemes have been developed \cite{lubich2014projector, ceruti2022unconventional, nobile2026high} to solve these equations with desirable accuracy and robustness properties (e.g., against small singular values if the rank is over-approximated by the user). For a small enough time step, the time-discrete DLRA solution applies instantaneously the truncated SVD of the true solution \cite{feppon2018geometric}, without ever computing an expensive SVD. Symplectic formulations have also been developed to preserve the Hamiltonian structure of wave equations. To name a few, the DLRA has been successfully applied to the Schrödinger wave equation \cite{ceruti2022rank}, the shallow water equation \cite{pagliantini2021dynamical} or the 2D stochastic wave equation \cite{musharbash2020symplectic}. To the best of the author's knowledge, however, the DLRA has not yet been applied to SHM-related elastic guided wave problems. In addition, despite the impressive above-mentioned progresses, its computational cost remains hardly actionable for the real-time constraints of SHM applications. \\

In this work, we attempt to reconcile the efficiency of the off-line/on-line logic, well established for projection-based ROM techniques applied to linear dissipative systems, with Hamiltonian wave equations. The key idea is to defer the expensive computation of the time-dependent basis to an off-line phase. More precisely, a symplectic DLRA is applied to simultaneously solve the elastodynamic equations for a discrete set of training parameters $\mathcal{D}^{\text{training}} \subset \mathbb{R}^P$, where $P$ is the number of parameters of interest. In this study, the latter correspond to the diameter of the transducer exciting the structure and to the central frequency of the loading (hence $P=2$). This yields a global-in-parameter but time-dependent basis of low rank $r$ (typically $10-30$):
\begin{equation}
    U(t; \mathcal{D}^{\text{training}})
\end{equation}
satisfying three key properties at all times: it is orthonormal, symplectic, and dynamically orthogonal (meaning its time derivative is orthogonal to its range.). The on-line phase then consists in solving a cheap Hamiltonian reduced flow at a discrete query parameter set $\mathcal{D}^{\text{query}} \subset \mathbb{R}^P$. Finally, the quantities of interest are efficiently obtained with the pre-computed basis. For instance, the high-dimensional full wave field $Y \in \mathbb{R}^{N \times m}$, with $N$ the huge number of spatial grid points and $m$ the number of query parameters, is reconstructed at a user-selected time $T$ as:
\begin{equation}
    Y(T;\mathcal{D}^{\text{query}}) \approx U(T; \mathcal{D}^{\text{training}}) Y_r(T;\mathcal{D}^{\text{query}}),
\end{equation}
where $Y_r(T;\mathcal{D}^{\text{query}}) \in \mathbb{R}^{r \times m}$ is the low-dimensional matrix solution of a reduced Hamiltonian flow evaluated at the query parameters. We validate this approach on a 2D plane strain elasticity problem with dispersive guided waves, demonstrating high compression ratios, important speedups, and negligible energy fluctuation. \\ 

The paper is organized as follows.
In section \ref{sec:preliminaries}, we introduce the main concepts and notations used throughout the paper, ensuring it is self-contained for readers unfamiliar with structure-preserving model order reduction.
In section \ref{sec:complexification}, we derive a complex phase-space formulation of the semi-discretized elastodynamic equations, yielding a unitary flow. 
In section \ref{sec:challenges_guided_waves}, we discuss the challenges posed by elastic guided waves, including their dispersive nature and the computational demands for resolving such phenomena.
In section \ref{sec:symplectic_DLRA}, we apply the Dynamical Low Rank Approximation (DLRA) method to our system. We derive an analytical solution to the nonlinear left basis evolution equation in a simplified parametric setting.
In section \ref{subsec:offline_online_strategy}, we detail our off-line/on-line strategy to efficiently handle parametric studies. 
In section \ref{sec:error_analysis}, we provide an error analysis to assess the accuracy of the method.
Finally, in section \ref{sec:numerical_application} we demonstrate the numerical performances through a transducer design problem for a guided wave-based SHM application.

\section{Preliminaries}
\label{sec:preliminaries}

\subsection{Matrix manifolds}
\label{subsec:matrix_manifolds}
We denote by $\mathbb{R}^{n \times m}$ and $\mathbb{C}^{n \times m}$ the spaces of real and complex $n \times m$ matrices, respectively. The following notations are adopted for the matrix manifolds used in this work:

\begin{itemize}
	\item \textbf{Stiefel manifold}: 
    $\mathcal{V}_r(\mathbb{R}^n) = \{ X \in \mathbb{R}^{n \times r} \mid X^T X = I_r \}$, where $I_r$ is the $r \times r$ identity matrix. This is the set of orthonormal $n \times r$ matrices, referred as reduced bases in the reduced order modeling terminology, with typically $r \ll n$ in the numerical applications.
    \item \textbf{Complex Stiefel manifold}: 
    $\mathcal{V}_r(\mathbb{C}^n) = \{ X \in \mathbb{C}^{n \times r} \mid X^H X = I_r \}$, where $X^H$ is the Hermitian transpose of $X$.
    \item \textbf{Orthogonal matrices}: $\text{O}(n) = \{Q \in \mathbb{R}^{n \times n} \mid Q^T Q = Q Q^T = I_n\}$
    \item \textbf{Unitary matrices}: $\text{U}(n) = \{Q \in \mathbb{C}^{n \times n} \mid Q^H Q = QQ^H = I_n\}$
    \item \textbf{Fixed rank matrices}: $\mathcal{M}_r = \{ X \in \mathbb{K}^{n \times m}, \quad \text{rank}(X) = r \}$, with $r \leq \min(n,m)$ and $\mathbb{K} = \mathbb{R}$ or $\mathbb{C}$.
    \item \textbf{Symplectic matrices}: $\text{Sp}(2n, 2r) = \{ X \in \mathbb{R}^{2n \times 2r} \mid X^T J_{2n} X = J_{2r} \}$, where $J_{2n} \in \mathbb{R}^{2n \times 2n}$ is the canonical skew-symmetric matrix:
    \[
    J_{2n} = \begin{bmatrix}
    0 & I_n \\
    -I_n & 0
    \end{bmatrix}.
    \]
    It satisfies $J_{2n}^T = - J_{2n} = J_{2n}^{-1}$. We also define the $2r \times 2n$ symplectic pseudo-inverse of $X$:
    \[
    X^+=J_{2r}^T X^T J_{2n}.
    \]
    If $X$ is symplectic it satisfies:
    \[
    X^+X = I_{2r}.
    \]
    \item \textbf{Ortho-symplectic matrices}: $\mathbb{S}(2n, 2r) = \mathcal{V}_{2r}(\mathbb{R}^{2n}) \cap \text{Sp}(2n,2r)$. A $2n \times 2r$ ortho-symplectic matrix has the following block structure \cite{peng2016symplectic}:
    \[
    \begin{bmatrix}
    \Phi & -\Psi \\
    \Psi & \Phi
    \end{bmatrix}
    \]
    where $\Phi, \Psi \in \mathbb{R}^{n \times r}$ satisfy $\Phi^T \Phi + \Psi^T \Psi = I_r$ and $\Phi^T\Psi = \Psi^T \Phi$. It commutes with the canonical skew-symmetrix matrix:
    \[
    X J_{2r} = J_{2n} X,
    \]
    hence its symplectic pseudo-inverse reduces to its transpose:
    \[
    X^+ = X^T.
    \]
    
\end{itemize}

\subsection{Hamiltonian wave equation}
\label{subsec:hamiltonian_wave_equation} Let $y = [q_1,\cdots,q_n,p_1,\cdots,p_n] \in \mathbb{R}^{2n}$ denote a phase-space vector of the studied system, with generalized coordinates $q_i$ and generalized momenta $p_i$. Hamiltonian wave equations (discretized in space) take the following form in an appropriate choice of phase-space coordinates:
\begin{equation}
    \dot y(t) = J_{2n} \nabla_y \mathcal{H}(y(t)),
    \label{eq:general_hamiltonian}
\end{equation}
with initial conditions $y(0) = y_0$. Here, $\mathcal{H}(y)$ is the Hamiltonian functional, $\nabla_y \mathcal{H}$ is the gradient of $\mathcal{H}$, and $J_{2n}$ is the canonical skew-symmetric matrix. Such systems conserve the Hamiltonian energy, that is $\mathcal{H}(y(t)) = \mathcal{H}(y(0))$ for all $t \geq 0$. Indeed: 
\[
\frac{d\mathcal{H}}{dt}(y(t)) = \nabla_y \mathcal{H}(y(t))^T \dot{y}(t) = \nabla_y \mathcal{H}(y(t))^T J_{2n} \nabla_y \mathcal{H}(y(t)) = 0,
\]
since $J_{2n}$ is skew-symmetric.\\

In this work, we consider quadratic Hamiltonians $\mathcal{H}(y) = \frac{1}{2} y^T H y$ and a forcing term exciting spatially localized waves. This results in a linear ordinary differential equation with a non-homogeneous part of the form:
\begin{equation}
\dot y(t) = J_{2n} H y(t) + b(t)
\label{eq:linear_hamiltonian_forcing}
\end{equation}
where $H \in \mathbb{R}^{2n \times 2n}$ is symmetric ($H^T = H$).

\subsection{Symplectic time marching schemes}
\label{subsec:symplectic_time_marching}

Preserving the Hamiltonian structure (\ref{eq:general_hamiltonian}), or (\ref{eq:linear_hamiltonian_forcing}) when the loading has vanished ($b=0$), at the time discrete level requires symplectic time marching schemes. We briefly describe below the structure-preserving integrator used in this study. For further details we refer the reader to \cite{hairer2006geometric}. We denote by $\Delta t$ a uniform time step. 

\begin{itemize}
    \item \textbf{Störmer-Verlet method:}
    Without forcing and for a separable Hamiltonian of the form $\mathcal{H}(q,p)=T(p)+V(q)$, this is an explicit, second-order accurate and symplectic integrator. Including a quadrature rule on a time-dependent forcing $b(t) = [0;f(t)]$, the update rule used in this work is:
\[ 
\begin{aligned} p_{n+\frac12} &= p_n +\frac{\Delta t}{2} \left( -\nabla_qV(q_n) +f_{n} \right), \\[0.5em]
 q_{n+1} &= q_n +\Delta t \nabla_pT(p_{n+\frac12}), 
\\[0.5em] p_{n+1} &= p_{n+\frac12} +\frac{\Delta t}{2} \left( -\nabla_qV(q_{n+1}) +f_{n+1} \right). 
\end{aligned} 
\]
\end{itemize}

\subsection{Structure-preserving reduced order modeling}
\label{subsec:structure_preserving_rom}

Using a static basis, a structure-preserving ROM approximates the full-order solution \( y(t) \in \mathbb{R}^{2n} \) as:
\[
y(t) \approx A y_{2r}(t),
\]
where \( A \in \text{Sp}(2n, 2r) \) is a symplectic reduced basis and \( y_{2r}(t) \in \mathbb{R}^{2r} \) is the reduced state, with $r \ll n$.

\begin{itemize}
    \item \textbf{Proper Symplectic Decomposition (PSD):}
    Let $Y = [y(t_1),\cdots,y(t_K)] \in \mathbb{R}^{2n \times K}$ be a matrix of $K$ full-order solution snapshots at times $t_1, \cdots, t_K$. The PSD seeks to compute an optimal $2n \times 2r$ (static) symplectic basis by solving:
    \[
    \min_{A \in \text{Sp}(2n,2r)} \| Y - AA^+ Y \|_F
    \]
    where $\| . \|_F$ denotes the Frobenius norm. In contrast, the POD seeks an orthonormal bases $A \in \mathcal{V}_{2r}(\mathbb{R}^{2n})$ with the orthogonal projector $AA^T$.
    \item \textbf{Complex SVD method} \cite{peng2016symplectic}: Unlike the POD, there is no known solution to the PSD problem. There is one, however, by reducing the feasibility set to ortho-symplectic matrices $\mathbb{S}(2n,2r)$. The snapshots are first complexified as:
    \[ 
    z(t_k) = q(t_k) + i p(t_k)
    \]
    where $q(t_k), p(t_k) \in \mathbb{R}^n$ are the vectors of generalized coordinates and momenta at time $t_k$, respectively. A complex $n \times n$ snapshot matrix is defined as $Z = Q + iP$, with $Q_{i,k} = q_i(t_k)$ and $P_{i,k} = p_i(t_k)$. Then a Singular Value Decomposition (SVD) is computed as:
    \[ 
    Z = U \Sigma V^H.
    \]
    Let $U_r = U(:,1:r)$ denote the first $r$ left complex singular vectors of $Z$. Then, the ortho-symplectic solution of the PSD problem is obtained with a realification of $U_r$ defined as:
\begin{equation}
A = \begin{bmatrix}
    \Re(U_r) & -\Im(U_r) \\
    \Im(U_r) & \Re(U_r)
    \end{bmatrix},
    \label{eq:realification}
\end{equation}
where $\Re(\cdot)$ and $\Im(\cdot)$ denote the real and imaginary parts, respectively.
    \item \textbf{Symplectic Galerkin projection:}
    Let $A \in \text{Sp}(2n,2r)$ be a symplectic reduced basis. The reduced dynamics are obtained by obliquely projecting the full-order system with the symplectic projector $AA^+$:
    \[
    \dot{y}_{2r} = J_{2r} A^T H A y_{2r} + A^+ b,
    \]
    with initial conditions $y_{2r}(0) = A^+y_0$. The reduced matrix $A^T H A$ is symmetric. Whenever the forcing vanishes, the reduced flow is symplectic hence the reduced quadratic Hamiltonian $\tilde{\mathcal{H}}(y_{2r}) = \frac{1}{2} y_{2r}^T A^T H A y_{2r}$ is preserved.
\end{itemize}

\section{Complexified Hamiltonian formulation of elastodynamic equations}
\label{sec:complexification}

In this section, we derive a reformulation of the semi-discretized elastodynamic equations, tailored for a symplectic Dynamical Low-Rank Approximation (DLRA). We begin by expressing the standard second-order-in-time formulation, involving the mass and stiffness matrices, as an equivalent first-order Hamiltonian system. We then introduce a linear complexification of the phase-space variables, which transforms the real Hamiltonian system into a unitarily evolving complex-valued system. 

\subsection{Studied problem and scope}
\label{subsec:studied_problem}

We consider a conservative linear elastodynamic problem with a time-dependent spatial loading depending on $P$ parameters of interest $\mu \in \mathbb{R}^P$. After discretization in space, the semi-discretized parameterized equations write \cite{bathe2006finite}:
\begin{equation}
    M\ddot{u}(t) + Ku(t) = f(t;\mu)
    \label{eq:elastodynamics}
\end{equation}
with initial conditions $u(0) = u_0$ and $\dot{u}(0) = \dot{u}_0$. Here, $u \in \mathbb{R}^n$ denotes a vector of degrees of freedom (e.g., displacements), $M \in \mathbb{R}^{n \times n}$ is the symmetric positive definite mass matrix, and $K \in \mathbb{R}^{n \times n}$ is the symmetric positive stiffness matrix. The loading $f(t;\mu)$ is assumed of duration $T_L$  short enough to significantly excite elastic waves (this is made more precise in section \ref{sec:challenges_guided_waves}). We split the time interval in two parts: 
\[
    (0,T) = (0,T_L) \cup (T_L,T)
\]
where the first part $(0,T_L)$ is referred as the loading regime, with generally $T_L \ll T$, and the second part $(T_L,T)$ is referred as the wave propagation regime, conservative and of much larger duration. \\

Essentially, the proposed ROM framework adopts a distinct strategy for each regime. Namely, the loading regime will be tackled with a (standard) symplectic Galerkin projection method. The wave propagation regime, however, will be tackled with a symplectic DLRA formulation involving a time-dependent reduced basis. \\

Note that we only consider the influence of parameters on the forcing. This is mostly motivated by all the simplifications occurring in the DLRA framework, but also by the SHM transducer design problem studied in section \ref{sec:numerical_application}. A more general parametric dependence, involving $M(\mu)$ and $K(\mu)$, is left to future work. 

\subsection{Full-order Hamiltonian model}
\label{subsec:hamiltonian_formulation}

We consider the usual phase-space vector $y=[q;p] \in \mathbb{R}^{2n}$, with $q = u$ (generalized coordinates) and $p=M\dot{u}$ (generalized momenta). Equation (\ref{eq:elastodynamics}) is then recast in a canonical Hamiltonian form \cite{buchfink2019symplectic}:
\begin{equation}
    \dot{y}(t) = J_{2n} H y(t) + b(t;\mu),
    \label{eq:hamiltonian_formulation}
\end{equation}
with the symmetric matrix:
\begin{equation}
    H = \begin{bmatrix}
    K & 0 \\
    0 & M^{-1}
\end{bmatrix},    
\label{eq:hamiltonian_matrix}
\end{equation}
the extended forcing:
\[
b(t;\mu) = \begin{bmatrix}
0 \\
f(t;\mu)
\end{bmatrix},
\]
and the initial condition $y(0)=[u_0;M\dot{u}_0]$. In the wave propagation regime, namely for $t \geq T_L$, the system's energy is conserved and amounts to $\mathcal{H}(y(t)) = \frac{1}{2} y_L^T H y_L$, where $y_L = [u(T_L);M\dot{u}(T_L)]$ is the phase-space vector at the end of the loading regime. \\

In this study, full-order simulations (that is without reduced order modeling) are obtained by advancing in time equation (\ref{eq:hamiltonian_formulation}) with the Störmer-Verlet scheme presented in section \ref{subsec:symplectic_time_marching}. For the spatial discretization, see section \ref{subsec:discretization}, we use the Spectral Element Method (SEM) with the so-called GLL quadrature. It yields a diagonal mass matrix tailored for explicit time marching schemes.

\subsection{Complexification}
\label{subsec:complex_phase_space}

The Hamiltonian system \eqref{eq:hamiltonian_formulation} provides a natural starting point for structure-preserving model order reduction. However, the particular form of the Hamiltonian matrix \eqref{eq:hamiltonian_matrix} allows a further simplification. Specifically, the dynamics can be recast as a unitary flow in a complex phase-space. \\

%Lemma 4.3 Peng and Mohseni: complex stiefel isomorphic to ortho-symplectic matrices. According to the Lemma, symplectic galerkin projection with ortho symplectic basis in real phase-space is equivalent to standard galerkin proejction in orthonormal complex basis in complex phase-space. Motivating symplectic terminology of this work even if standard orthogonal projections used in complex phase-space.\\

Let $u(t)\in\mathbb{R}^n$ solve \eqref{eq:elastodynamics} and define the complex phase-space vector:
\[
z(t)=Au(t)+iB\dot{u}(t),
\]
where $A,B\in\mathbb{R}^{n\times n}$ are nonsingular matrices. The Hamiltonian formulation of section~\ref{subsec:hamiltonian_formulation} corresponds to $A=I_n$ and $B=M$, which yields a symplectic flow in a real phase-space. However, suitable pairs $(A,B)$ can be chosen so that the dynamics becomes unitary in the complex phase-space. The following proposition introduces a canonical choice.

\begin{proposition}[Unitary complex flow]
Assume $K$ is symmetric positive definite (no rigid body modes) and let:
\begin{equation*}
\Omega^2 = M^{-1/2} K M^{-1/2}
\end{equation*}
be the mass-normalized stiffness matrix. Then $\Omega := (\Omega^2)^{1/2}$ is well-defined, symmetric positive definite, and invertible. Define:
\begin{equation*}
A = M^{1/2}, \qquad B = \Omega^{-1} M^{1/2}.
\label{eq:A_B}
\end{equation*}

Then the evolution of $z = A u + i B \dot{u}$ decouples from its complex conjugate $\overline{z}$ and satisfies:
\begin{equation}
\dot z(t)=-i\Omega z(t)+i g(t;\mu), \quad g(t;\mu) = \Omega^{-1}M^{-1/2}f(t;\mu).
\label{eq:complex_flow}
\end{equation}
The flow is unitary in the wave propagation regime:
\begin{equation*}
\Vert z(t) \Vert = \Vert z(T_L)\Vert,
\qquad t\ge T_L.
\end{equation*}
\label{prop:canonical_complexification}
\end{proposition}
\begin{proof}
Let $q=M^{1/2} u$. Then (\ref{eq:elastodynamics}) becomes $\ddot{q} + \Omega^2 q = M^{-1/2}f$. Let $p=\Omega^{-1}\dot{q}$. Then a direct computation of the time derivative of $z= q + i p$ gives the result.
\end{proof}

Some comments are in order. This complex phase-space representation is nothing more than a first-order in time formulation of the well-known modal superposition principle in structural dynamics \cite{bathe2006finite}. The square root $\Omega$ of the mass-normalized stiffness matrix can be obtained from standard solvers as follows. First, compute the mass-normalized modal basis $\Phi \in \mathbb{R}^{n\times n}$ solving the generalized eigenvalue problem:
\[
K\Phi = M \Phi \Lambda \quad \text{s.t.} \quad \Phi^T M \Phi = I_n,
\]
where $\Lambda$ is a diagonal matrix with (here strictly) positive entries. Then, using $M^T = M$ and $\Phi^{-1} = \Phi^T M$, we obtain:
\begin{equation}
    \Omega = \Psi \sqrt{\Lambda} \Psi^T \quad \text{with} \quad \Psi = M^{1/2} \Phi \quad (\Psi^T\Psi = I_n).
    \label{eq:computation_omega}
\end{equation}
In the real phase-space defined with the generalized coordinates $q= Au$ and momenta $p=B\dot{u}$, with $(A,B)$ given by the proposition above, the resulting Hamiltonian matrix takes the simpler block diagnonal form $H=\text{blkdiag}(\Omega, \Omega)$. Hence it commutes with the canonical skew-symmetric matrix, that is $J_{2n}H = H J_{2n}$, just as the left and right multiplications of $\Omega$ by the imaginary number $-i$  in (\ref{eq:complex_flow}) are trivially equal. While the real and complex representations are equivalent, we proceed with the simpler algebra in the complex phase-space and the standard properties of unitary flows. 

\subsection{Wave propagation regime}
\label{subsec:unitary_matrix_flow}
We instantiate the unitary complex flow (\ref{eq:complex_flow}) at a discrete set of $m$ parameter values $\mu \in \mathcal{D} \subset \mathbb{R}^P$. For clarity, we keep the notation $T_L$ to denote the maximum loading duration observed for $\mu \in \mathcal{D}$. In the wave propagation regime, this yields a unitary matrix flow for $t\geq T_L$:
\begin{equation}
    \dot{Z}(t) = -i \Omega Z(t),
    \label{eq:matrix_complex_flow}
\end{equation}
where $Z(t) = [z(t;\mu_1),\cdots,z(t;\mu_m)] \in \mathbb{C}^{n \times m}$ stacks the (unrelated and decoupled) solutions of the unitary complex flow (\ref{eq:complex_flow}) at each parameter value. The initial condition is:
\begin{equation}
Z(T_L) = M^{1/2} W(T_L) + i \Omega^{-1}M^{1/2}\dot{W}(T_L),
\label{eq:initial_condition_Z}
\end{equation}
with $W(T_L) \in \mathbb{R}^{n\times m}$ the stacked solutions of (\ref{eq:elastodynamics}) at the end of the loading regime (we used the notation $W$ to avoid the conflict with the left basis $U$ introduced later). \\

Equation (\ref{eq:matrix_complex_flow}) represents a continuous rotation of the initial data $Z(T_L)$ with instantaneous generator the (skew-Hermitian) matrix $i\Omega$. The flow can be solved analytically using the basic theory of linear differential equations as:
\begin{equation}
    Z(t) = e^{-i\Omega t}Z(T_L) = \Psi e^{-i \sqrt{\Lambda} \tau} \Psi^T Z(T_L), \quad t \geq T_L
    \label{eq:wave_exact_solution}
\end{equation}
with $\tau = t-T_L$, $\Psi$ and $\sqrt{\Lambda}$ from the spectral decomposition (\ref{eq:computation_omega}). \\

This paper addresses the challenge of approximating the solution (\ref{eq:wave_exact_solution}), involving a full static modal basis, with a low-dimensional, time-dependent basis tailored for elastic guided wave problems. The efficiency of the proposed approach is benchmarked in the numerical section against a standard modal truncation technique.

\section{Challenges posed by elastic guided waves}
\label{sec:challenges_guided_waves}

\subsection{Dispersion}
\label{subsec:dispersion}

Elastic waves in unbounded, homogeneous, and isotropic media admit two fundamental modes: longitudinal P-waves and transverse S-waves. However, the presence of boundaries or interfaces gives rise to a far richer family of solutions called guided wave modes. Of interest to this study are the so-called Lamb wave modes, which propagate in plates having two free opposite surfaces. Their velocities are highly frequency-dependent. Specifically, a dispersion relation establishes the condition for a Lamb wave mode to exist with a circular frequency $\omega$ and a wavenumber $k$ (\cite{su2009identification}, chapter 2):
\begin{equation}
\frac{\tanh (\beta d)}{\tanh (\alpha d)} = \left[ \frac{4\alpha \beta k^2}{(k^2 + \beta^2)^2} \right]^m
\label{eq:dispersion_relation}
\end{equation}
where $d$ is the half-thickness of the plate, $\alpha$ and $\beta$ are functions of $\omega$ and $k$, and $m = \pm 1$. The through-thickness profile is either symmetric for solutions obtained with $m=+1$, or antisymmetric for those obtained with $m=-1$. \\

In the $(\omega,k)$ plane, the solutions of the nonlinear relation (\ref{eq:dispersion_relation}) form branches, each corresponding to a distinct Lamb wave mode. Within a given frequency band, multiple branches may coexist, leading to complex interactions such as mode conversion at boundaries or interfaces. These interactions correspond to an exchange of energy between wave modes, further complicating the dynamics. 

\subsection{Discretization}
\label{subsec:discretization}

Accurately simulating such phenomena requires fine temporal and spatial resolutions to capture the highest frequencies and shortest wavelengths of interest.  Moreover, the time-marching scheme must minimize numerical dissipation and preserve energy over long times. We address these requirements in our full-order simulations by combining the Spectral Element Method (SEM) and the Störmer-Verlet scheme. Indeed, the SEM leverages high-order finite elements with established accuracy for wave propagation simulations \cite{komatitsch1999spectral}. In practice, the experience with such simulations is to use the well-known critical time step of the central difference method (CDM), and 12 spatial grid points along the shortest wavelength of interest \cite{duczek2014numerical}. Specifically, we chose:
\begin{equation}
\Delta t \leq \Delta t_\textrm{CDM} = \frac{2}{\sqrt{\rho(M^{-1}K)}}, \quad \Delta x \leq \frac{p}{12} \lambda_{\min},
\end{equation}
where $\rho(\cdot)$ denotes the spectral radius, $p$ is the polynomial order of the SEM elements, and $\lambda_{\textrm{min}}$ is the target wavelength. 

\subsection{Limits of a modal basis}
\label{subsec:limits_modal_basis}

The contrast between the complex physical behavior of guided waves and the seamingly simple analytical solution \eqref{eq:wave_exact_solution} is only apparent. It disappears as one observes that a spatially localized moving wave packet exhibits a broad frequency spectrum, activating the superposition of many global vibration modes. The spatial localization at a given time is a consequence of their interference pattern, which concentrates the wave field within a limited spatial region and attenuates it elsewhere. For example, in the numerical application of section~\ref{sec:numerical_application}, the loading's frequency spectrum overlaps with a region of $\Omega$'s spectrum where hundreds of vibration modes ($\sim500$) are active. For more complex structures (e.g., heterogeneous or anisotropic media), this number can increase significantly, rendering direct modal representations computationally intractable for large-scale simulations. \\

To overcome these challenges, we compute a Dynamical Low Rank Approximation (DLRA) of the solution. The high dimensional (static) modal subspace is replaced by a low-dimensional, but time-dependent, symplectic subspace. The spectral complexity of elastic guided waves is then retained without the prohibitive cost of a full (or at least very large) modal representation.

\section{Symplectic Dynamical Low Rank Approximation (DLRA)}
\label{sec:symplectic_DLRA}

In this section we recall the basics of the DLRA method before applying it to the linear matrix differential equation (\ref{eq:matrix_complex_flow}). We derive an analytical solution of the nonlinear left basis evolution equation, which motivates the proposed off-line/on-line reduced order modeling strategy.

\subsection{DLRA equations}
\label{subsec:DLRA_equations}

\begin{figure}[t]
    \centering
    \includegraphics[width=0.7\textwidth]{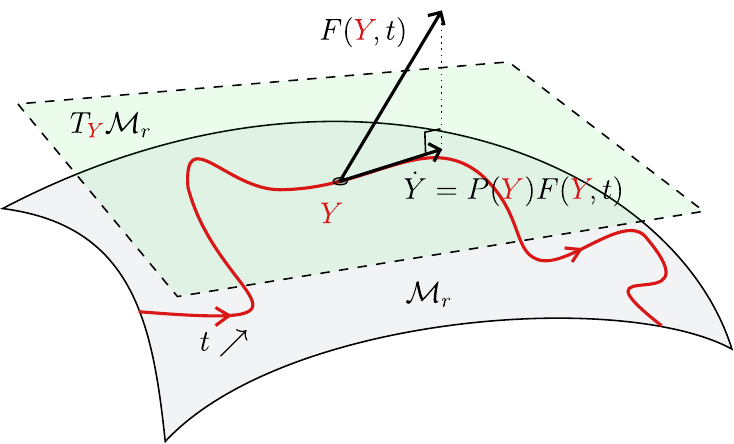}
    \caption{The DLRA advances in time a rank-$r$ approximation $Y(t)$ of the solution of a matrix differential equation $\dot{X} = F(X,t)$. It proceeds with an orthogonal projection of the field $F(Y(t),t)$ onto the tangent space of the rank-$r$ matrix manifold $\mathcal{M}_r$ at the current approximation $Y(t) \in \mathcal{M}_r$. Figure adapted from \cite{feppon2018geometric}.}
    \label{fig:DLRA}
\end{figure}

The DLRA \cite{koch2007dynamical} seeks an SVD-like factorization of the solution of a matrix differential equation $\dot{Y} = F(Y,t)$, where $Y(t)$ is a huge time-dependent $n \times m$ (real or complex) matrix and $F$ is a (linear or nonlinear) function of $Y$ and time. We proceed with the complex setting as our matrix differential equation of interest is the unitary complex flow (\ref{eq:matrix_complex_flow}). Specifically, the objective is to find a low rank $r$ approximation of the true solution, with $r \ll \min(n,m)$, of the form:
\begin{equation}
    Y(t) = U(t) S(t) V(t)^H
    \label{eq:SVD_factorization}
\end{equation}
where $U(t) \in \mathcal{V}_r(\mathbb{C}^n)$ is the (time-dependent) $n \times r$ left basis, $V(t) \in \mathcal{V}_r(\mathbb{C}^m)$ is the $m \times r$ right basis, and $S(t) \in \mathbb{C}^{r \times r}$ is nonsingular. Note that the matrix $S(t)$ is not enforced to be diagonal as in a standard SVD. \\

The factorization (\ref{eq:SVD_factorization}) defines a parameterization $(U,S,V)$ of $Y$ in the manifold of fixed rank $r$ matrices $\mathcal{M}_r$. It is not unique. Any triplet $(UQ, Q^H S Q, VQ)$, with $Q \in U(r)$ unitary, yields an equally valid parameterization. This freedom is leveraged in the method as described below. \\

A basic observation is that if we seek an approximation $Y(t) \in \mathcal{M}_r$ for all times, then its time derivative $\dot{Y}(t)$ must belong to the tangent space of $\mathcal{M}_r$ for all times as well. This writes: $\dot{Y}(t) \in T_{Y(t)} \mathcal{M}_r$. In fact, the DLRA proceeds by minimizing the distance of the residual $\dot{Y} - F(Y,t)$ to $T_Y \mathcal{M}_r$ from an initial data $Y_0 \in \mathcal{M}_r$. The gauge freedom in the SVD-like parameterization of elements of $\mathcal{M}_r$ is leveraged to define, instead, a unique parameterization in its tangent space. Specifically, the DLRA enforces the so-called dynamically-orthogonal conditions:
\begin{equation}
    U^H \dot{U} = V^H \dot{V} = 0,
    \label{eq:dynamically_orthogonal}
\end{equation}
for all times. The factors in (\ref{eq:SVD_factorization}) are determined as the solution of the following system of differential equations:  
\begin{align}
    \dot{S} &= U^H F(Y,t) V \label{eq:S_equation}\\
    \dot{U} &= P_U^{\perp} F(Y,t) S^{-1} \label{eq:U_equation}\\
    \dot{V} &= P_V^{\perp} F(Y,t)^H U S^{-H} \label{eq:V_equation}
\end{align}
with initial conditions $U_0,S_0,V_0$ obtained by performing a truncated SVD of the initial data at rank $r$. These equations are structure-preserving at the continuous time level. For instance, one observes that $\frac{d}{dt}U^H U = \dot{U}^H U + U^H \dot{U} = 0$ since $\dot{U}$ given by (\ref{eq:U_equation}) is orthogonal to the range of $U$. Thus $U(t)^H U(t) = U_0^H U_0 = I_r$ if $U_0 \in \mathcal{V}_r(\mathbb{C}^n)$. In practice, dedicated time marching schemes must be used to reliably integrate these equations. Impressive progresses have been accomplished over the past decades to derive algorithms that are robust to small singular values (e.g., if the rank $r$ is over-approximated \cite{lubich2014projector, ceruti2022unconventional}), accurate (high order schemes \cite{nobile2026high}), and energy-preserving \cite{pagliantini2021dynamical} if the studied matrix differential equation is symplectic .

\subsection{Analytical solution of the left basis}
\label{subsec:analytical_left_basis}

We now apply the DLRA to (\ref{eq:matrix_complex_flow}). First, we remark that a homogeneous linear flow is rank-preserving. As such, if the initial data $Z(T_L)$ is of (hopefully low) rank $r$:
\begin{equation}
    Z(T_L) = U_r S_r V_r^H \in \mathcal{M}_r,
    \label{eq:initial_SVD}
\end{equation}
then so is the true solution in the wave propagation regime:
\begin{equation*}
    Z(t) \in \mathcal{M}_r, \quad t \geq T_L.
\end{equation*}

Therefore, the rank of the solution is only expected to grow in the loading regime from null initial conditions, motivating the (less structured) model order reduction strategy employed in the time interval $(0,T_L)$, see section \ref{subsec:ROM_loading_regime}. \\

Next, we observe that equations (\ref{eq:S_equation}) to (\ref{eq:V_equation}) considerably simplify for the studied linear flow. For $t\geq T_L$:
\begin{align}
    \dot{S} &= -iU^H \Omega U S \label{eq:S_equation_linear}\\
    \dot{U} &= -iP_U^{\perp} \Omega U \label{eq:U_equation_linear}\\
    \dot{V} &= 0 \label{eq:V_equation_linear}
\end{align}
with the initial conditions $U(T_L) = U_r, S(T_L) = S_r$ and $V(T_L) = V_r$ from (\ref{eq:initial_SVD}). \\

In particular, the inverse of the $S$ factor in equations (\ref{eq:U_equation})-(\ref{eq:V_equation}) cancels out, and the right basis $V$ is constant in time (it would not if the flow had the more general Sylvester form $\dot{X}=AX + XB$ encountered below). The left basis equation is nonlinear as $P_U^{\perp}=I-UU^H$. The following lemma considerably simplifies the analysis.

\begin{lemma}
The matrix $U^H \Omega U$ is constant along any solution of
\eqref{eq:U_equation_linear}.
\label{lemma:invariant}
\end{lemma}
\begin{proof}
First, we write equation (\ref{eq:U_equation_linear}) as:
\begin{equation*}
    \dot{U} = -i \Omega U + i U(U^H \Omega U).
\end{equation*}
Then a direct calculation using $\Omega^H = \Omega$ gives:
\begin{align*}
    \frac{d}{dt} U^H \Omega U &= \dot{U}^H \Omega U + U^H \Omega \dot{U} \\
    &= (i U^H \Omega -i U^H \Omega U U^H) \Omega U + U^H \Omega (-i \Omega U + i UU^H \Omega U) \\
    &= i(U^H \Omega^2 U - U^H \Omega U U^H\Omega U - U^H \Omega^2 U  + U^H \Omega UU^H \Omega U) \\
    &= 0.
\end{align*}
\end{proof}

Following this lemma, we define the Hermitian matrix:
\begin{equation}
    \Omega_r = U_r^H \Omega U_r,
    \label{eq:Omega_r}
\end{equation}
which absorbs the nonlinearity of equation (\ref{eq:U_equation_linear}), affecting only the behavior with respect to the initial condition $U_r$. The exact solution of the nonlinear left basis evolution is given below.
\begin{theorem}[Left basis formula]
Let $U(t) \in \mathcal{V}_r(\mathbb{C}^n)$ solve (\ref{eq:U_equation_linear}) and assume the flow starts at $U(T_L)=U_r$. Let $\Omega = \Psi \sqrt{\Lambda}\, \Psi^H$ be the spectral decomposition of $\Omega$ defined in (\ref{eq:computation_omega}). Let
also $\Omega_r = \Psi_r \sqrt{\Lambda_r}\, \Psi_r^H$ be the spectral decomposition of the reduced matrix introduced in (\ref{eq:Omega_r}). Then $U(t)$ is explicitly given for $t \geq T_L$ by:
\begin{equation}
    U(t) = \Psi e^{-i \sqrt{\Lambda} \tau} \Psi^H U_r \Psi_r e^{i \sqrt{\Lambda_r} \tau} \Psi_r^H,
    \label{eq:left_basis_solution}
\end{equation}
with $\tau = t-T_L$.
\label{prop:left_basis_solution}
\end{theorem}
\begin{proof}
Since $\Omega_r = U^H \Omega U$ is constant by lemma \ref{lemma:invariant}, the solution of the Sylvester form equation $\dot{U} = -i \Omega U + i U\Omega_r$ for $t\geq T_L$ with initial data $U_r$, is given by $U(t) = e^{-i \Omega \tau} U_r e^{i \Omega_r \tau}$. The result follows with the spectral decompositions of $\Omega$ and $\Omega_r$.
\end{proof}

The left basis formula (\ref{eq:left_basis_solution}) plays a central role in the proposed framework. It encodes the rich wave dynamics into a $r$-dimensional, time-dependent, symplectic subspace. Geometrically, this result reveals a double rotational structure. A high-dimensional, time-dependent stiff rotation $e^{-i \Omega \tau}$ acts on the subspace spanned by $U_r$ ; and a low-dimensional rotation $e^{i \Omega_r \tau}$ dynamically moves the coordinate system within that subspace. The simplification of the studied setting, see section \ref{subsec:studied_problem}, is that only the low-dimensional part involving $U_r$ depends on the parameters. \\

We note that another time-dependent, orthonormal, rank-$r$ basis $\tilde{U}(t)$ can be obtained directly from (\ref{eq:wave_exact_solution}) and (\ref{eq:initial_SVD}) as $\tilde{U}(t) = e^{-i\Omega \tau} U_r$. This basis provides an equally valid rank-$r$ parameterization of $Z(t)$ in $\mathcal{M}_r$, illustrating the gauge freedom mentioned in section \ref{subsec:DLRA_equations}. However, it is not dynamically orthogonal ($\tilde{U}^H \dot{\tilde U} = -i \tilde{U}^H \Omega \tilde{U} \neq 0$), thus it does not strictly separate the basis motion from the reduced coefficients dynamics. In contrast, the dynamically orthogonal basis $U(t)$ satisfies the standard DLRA gauge condition and therefore provides a canonical representative among the infinitely many equivalent rank-$r$ factorizations of $Z(t)$.

\subsection{Reduced unitary flow}
\label{subsec:reduced_unitary_flow}

Since the right basis is constant in time, see equation (\ref{eq:V_equation_linear}), we define the reduced matrix:
\begin{equation*}
   Z_r(t) = S(t) V_r^H \in \mathbb{C}^{r \times m}
\end{equation*}
where $V_r \in \mathcal{V}_r(\mathbb{C}^m)$ is the initial right basis obtained from the truncated SVD of the initial data at rank $r$, see (\ref{eq:initial_SVD}). Its evolution equation is obtained by right multiplying equation (\ref{eq:S_equation_linear}) by $V_r^H$:
\begin{equation}
    \dot{Z}_r = - i \Omega_r Z_r
    \label{eq:reduced_matrix_solution}
\end{equation}
with initial condition $Z_r(T_L) = S_r V_r^H$ and $\Omega_r$ given by (\ref{eq:Omega_r}). Hence this reduced flow is unitary. Its exact solution, after spectral decomposition of $\Omega_r$ as in theorem \ref{prop:left_basis_solution}, is given for $t \geq T_L$ by:
\begin{equation}
    Z_r(t) = e^{-i \Omega_r \tau} Z_r(T_L) = \Psi_r e^{-i \sqrt{\Lambda_r} \tau} \Psi_r^H Z_r(T_L),
    \label{eq:reduced_matrix_solution}
\end{equation}

with $\tau = t-T_L$. We can check that the two factors (\ref{eq:left_basis_solution}) and (\ref{eq:reduced_matrix_solution}) of the DLRA solution recombine the solution (\ref{eq:wave_exact_solution}), provided that the initial data is of rank $r$:
\begin{align*}
    U(t)Z_r(t) &= e^{-i \Omega \tau}  U_r e^{i \Omega_r \tau}  e^{-i \Omega_r \tau} Z_r(T_L) \\
    &= e^{-i \Omega \tau}  U_r Z_r(T_L) \\
    &= e^{-i \Omega \tau} U_r S_r V_r^H \\
    &= e^{-i \Omega \tau} Z(T_L) \\
    &= Z(t).
\end{align*}

The non-trivial part was to factor the solution with a dynamically orthogonal rank-$r$ basis. The DLRA solution preserves the unitary structure in the complex phase-space, hence the symplecticity in the associated real phase-space, motivating the symplectic DLRA terminology adopted in this paper.

\section{Proposed off-line/on-line strategy}
\label{subsec:offline_online_strategy}

\begin{figure}[t]
    \centering
    \includegraphics[width=1\textwidth]{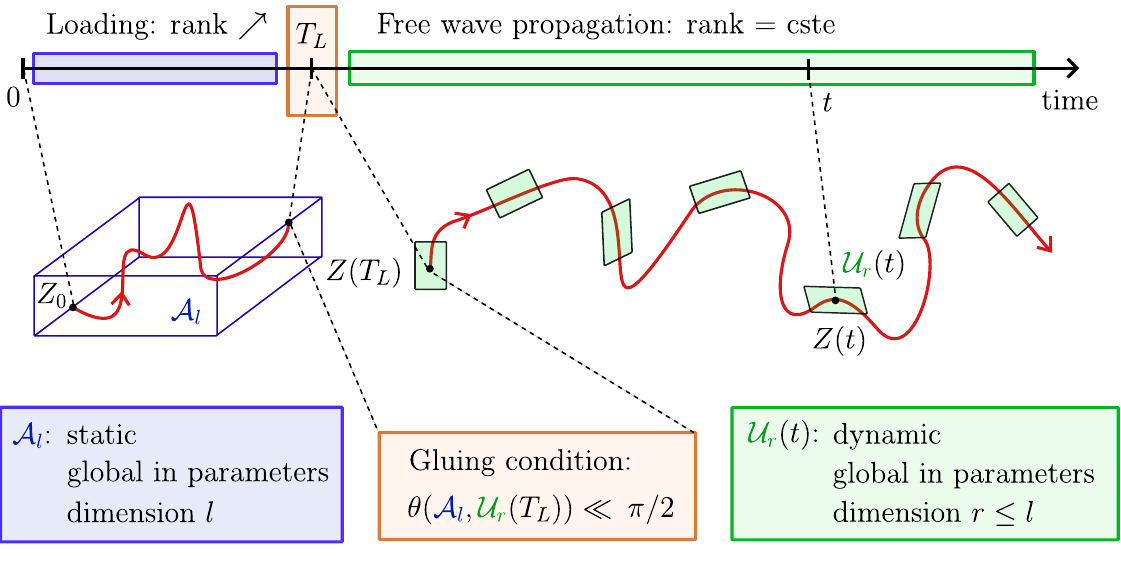}
    \caption{The complex matrix solution $Z(t) \in \mathbb{C}^{n \times m}$ is approximated with a distinct ROM method in each part of the time interval. In the loading regime, it is approximated with a symplectic Galerkin projection on a static $l$-dimensional subspace $\mathcal{A}_l$. In the wave propagation regime, it is approximated with a symplectic DLRA involving a dynamic $r$-dimensional subspace $\mathcal{U}_r(t)$. A gluing condition at the interface of the two regimes ensures consistency. }
    \label{fig:ROM_strategy}
\end{figure}

In this section, we describe our reduced order modeling framework which exploits the analytical solutions derived above to efficiently handle parametric studies. The method splits the simulation into the two distinct regimes: the loading regime $(0, T_L)$, where the solution rank grows, and the wave propagation regime $(t \geq T_L)$, where the rank remains constant. We consider a training set $\mathcal{D}^{\text{training}}$ for off-line computations and a query set $\mathcal{D}^{\text{query}}$ for on-line evaluations, with the goal of computing quantities of interest such as time series at sensor locations or full wavefield reconstructions. The off-line and on-line phases are summarized below in Algorithms~\ref{alg:offline} and \ref{alg:online}, respectively. The geometrical concepts of the proposed approach are illustrated on Figure \ref{fig:ROM_strategy}.

\subsection{Loading regime} 
\label{subsec:ROM_loading_regime}

On $(0,T_L)$, the full-order equations are:
\begin{equation}
\dot Z(t) = -i\Omega Z(t) + iG(t; \mathcal{D}^{\text{query}}),
\label{eq:FOM_loading_regime}
\end{equation}
where $G(t; \mathcal{D}^{\text{query}})$ is the forcing term matrix. The solution rank increases from null initial conditions, making a fixed-rank DLRA approximation impractical. Instead, we use the standard symplectic Galerkin method described in section \ref{subsec:structure_preserving_rom}. \\

\paragraph{Off-line phase} We prepare the reduced basis and spectral data needed for efficient on-line evaluation. Specifically, we compute a rank-$l$ basis $A_l \in \mathcal{V}_l(\mathbb{C}^n)$ with a complex SVD on the training snapshots $\{Y(t_k; \mathcal{D}^{\text{training}})\}_{t_k=0}^{T_L}$. Optionally, the time interval $(0, T_L)$ can be subsampled to reduce the computational effort. Then, we perform the spectral decomposition of the following reduced Hermitian matrix:
  \begin{equation}
      A_l^H \Omega A_l = \Psi_l \sqrt{\Lambda_l} \Psi_l^H.
      \label{eq:spectral_Omega_l_offline}
  \end{equation}

\paragraph{On-line phase} We perform a Galerkin projection of (\ref{eq:FOM_loading_regime}) onto the span of $A_l$, yielding:
  \begin{equation}
      \dot{Z}_l(t) = -i A_l^H \Omega A_l Z_l(t) + i G_l(t; \mathcal{D}^{\text{query}}), \quad t \leq T_L,
      \label{eq:ROM_loading_online}
  \end{equation}
with $G_l(t; \mathcal{D}^{\text{query}}) = A_l^H G(t; \mathcal{D}^{\text{query}})$. This equation is advanced in time in spectral coordinates as a Hamiltonian system with $(q,p)$ variables with a symplectic integrator. This yields the reduced solution $Z_l(T_L; \mathcal{D}^{\text{query}})$ at the end of the loading.

\subsection{Wave propagation} In $(T_L,T)$ the full-order equations are: 
\[
\dot{Z}(t) = - i \Omega Z(t).
\]
The solution rank is constant and dictated by the rank at $T_L$. The analytical DLRA solution is used to efficiently approximate the solution in this regime. \\

\paragraph{Off-line phase} We select a rank $r$ based on the complex SVD of the FOM matrix $Y(T_L; \mathcal{D}^{\text{training}})$ and on the gluing condition introduced later in section \ref{sec:error_analysis}. The (dominant) $r$ first left singular vectors define the initial left basis $U_r$ at $T_L$ in (\ref{eq:initial_SVD}). The analytical left basis:
\begin{equation}
U(t; \mathcal{D}^{\text{training}})
\label{eq:wave_training_basis}
\end{equation}
is then computed using (\ref{eq:left_basis_solution}). This amounts to storing $r$ full-order solutions, which can be prohibitive for large spatial grids and long time intervals. In practice, a masked basis selecting only the rows corresponding to sensor locations are extracted at user-selected times. Alternatively, if the quantity of interest is based on the wave field at a single time, then the basis needs only to be stored at this specific time. \\

\paragraph{On-line phase} The initial data for the reduced wave propagation regime are obtained by lifting $Z_l(T_L;\mathcal{D}^{\text{query}})$ with $A_l$ and projecting it onto the range of $U_r$:
\begin{equation}
Z_r(T_L;\mathcal{D}^{\text{query}}) = U_r^H A_l Z_l(T_L;\mathcal{D}^{\text{query}}).
\label{eq:wave_initial_condition}
\end{equation}
Note that the $r \times l$ matrix $U_r^H A_l$ can be pre-computed and stored. We then evolve $Z_r(t; \mathcal{D}^{\text{query}})$ analytically using (\ref{eq:reduced_matrix_solution}). The full solution is reconstructed as needed at any time $T \geq T_L$ from:
  \begin{equation}
      Z(T; \mathcal{D}^{\text{query}}) = U(T; \mathcal{D}^{\text{training}}) Z_r(T; \mathcal{D}^{\text{query}}).
      \label{eq:ROM_analytical}
  \end{equation}

\begin{algorithm}[H]
\caption{Off-line phase}
\label{alg:offline}
\begin{algorithmic}[1]
\STATE \textbf{Input:} Full-order snapshots $\{Y(t_k; \mathcal{D}^{\text{training}})\}_{t_k=0}^{T_L}$, operator $\Omega$.
\STATE \hrulefill
\STATE \textbf{Loading Regime:}
	\STATE Concatenate the snapshots $\{Y(t_k; \mathcal{D}^{\text{training}})\}_{t_k=0}^{T_L}$.
    \STATE Compute a rank-$l$ basis $A_l \in \mathcal{V}_l(\mathbb{C}^n)$ via complex SVD.
    \STATE Compute the reduced operator $\Omega_l = A_l^H \Omega A_l$.
    \STATE Perform the spectral decomposition $\Omega_l = \Psi_l \sqrt{\Lambda_l} \Psi_l^H$.
    \STATE \textbf{Output:} $A_l$, $\Psi_l$, $\sqrt{\Lambda_l}$.
    \STATE \hrulefill
\STATE \textbf{Wave Propagation Regime:}
    \STATE Compute a rank-$r$ basis $U_r \in \mathcal{V}_r(\mathbb{C}^n)$ via complex SVD of $Y(T_L; \mathcal{D}^{\text{training}})$.
    \STATE Compute the smallest singular value $\sigma_{\textrm{min}}$ of $U_r^H A_l$.
    \STATE Compute the largest principal angle between $\mathcal{U}_r$ and $\mathcal{A}_l$ as $\theta = \cos^{-1}(\sigma_{\textrm{min}})$.
    \STATE Check $\theta$ is below user tolerance or reduce rank $r$. 
    \STATE Compute the reduced operator $\Omega_r = U_r^H \Omega U_r$.
    \STATE Perform the spectral decomposition $\Omega_r = \Psi_r \sqrt{\Lambda_r} \Psi_r^H$.
    \STATE Compute the analytical left basis $U(t; \mathcal{D}^{\text{training}})$ with \eqref{eq:left_basis_solution} at query times.
    \STATE \textbf{Output:} $U_r$, $\Psi_r$, $ \sqrt{\Lambda_r}$, $U(t; \mathcal{D}^{\text{training}})$.
\end{algorithmic}
\end{algorithm}

\begin{algorithm}[H]
\caption{On-line phase}
\label{alg:online}
\begin{algorithmic}[1]
\STATE \textbf{Input:} Query set $\mathcal{D}^{\text{query}}$, $A_l$, $\Psi_l$, $\sqrt{\Lambda_l}$, $U_r$, $\Psi_r$, $ \sqrt{\Lambda_r}$, $U(t; \mathcal{D}^{\text{training}})$.
\STATE \hrulefill
\STATE \textbf{Loading Regime:}
	\STATE Build the forcing matrix $G(t; \mathcal{D}^{\text{query}})$.
    \STATE Solve the reduced flow \eqref{eq:ROM_loading_online} on $(0,T_L)$ with a structure-preserving scheme.
    \STATE \textbf{Output:} $Z_l(T_L; \mathcal{D}^{\text{query}})$.
    \STATE \hrulefill
\STATE \textbf{Wave Propagation Regime:}
    \STATE Compute the initial conditions $Z_r(T_L; \mathcal{D}^{\text{query}}) = U_r^H A_l Z_l(T_L; \mathcal{D}^{\text{query}})$.
    \STATE Propagate $Z_r(T; \mathcal{D}^{\text{query}})$ analytically using \eqref{eq:reduced_matrix_solution}.
    \STATE Reconstruct the full solution $Z(T; \mathcal{D}^{\text{query}}) = U(T; \mathcal{D}^{\text{training}}) Z_r(T; \mathcal{D}^{\text{query}})$.
    \STATE \textbf{Output:} $Z(T; \mathcal{D}^{\text{query}})$ (or other QoI).
\end{algorithmic}
\end{algorithm}

\section{Error analysis}
\label{sec:error_analysis}

We derive an error estimate by exploiting the unitary structure of the dynamics in the wave propagation regime. Let $Z(t)$ denote the exact full-order solution and:
\[
\tilde Z_l(t)=A_l Z_l(t), \quad \tilde Z_r(t)=U(t) Z_r(t)
\]
the lifted reduced solutions of the loading and wave propagation regimes, respectively. We recall from theorem \ref{prop:left_basis_solution}:
\[ 
U(t) = e^{-i \Omega \tau} U_r e^{i \Omega_r \tau}, \quad \tau = t -T_L, \quad \Omega_r = U_r^H \Omega U_r.
\] 

Let also $P_l=A_lA_l^H$ and $P_r=U_rU_r^H$ be the orthogonal projectors onto the ranges of $A_l$ and $U_r$, respectively.

\subsection{Loading regime}

On $(0,T_L)$, the equations are:
\[
\dot Z = -i\Omega Z + i G, 
\qquad
\dot{Z}_l = -i\Omega_l Z_l + i G_l,
\]
with $\Omega_l=A_l^H\Omega A_l$ and $G_l=A_l^H G$. Initial conditions are $Z(0) = Z_0$ and $Z_l(0) = A_l^H Z_0$ (possibly null). Let $E(t) = e^{-i \Omega t} - A_l e^{-i\Omega_l t} A_l^H$. The variation-of-constants formula yields for $t\leq T_L$:
\[
\|Z(t)-\tilde Z_l(t)\|
\;\le\;
 \varepsilon_{\mathrm{loading}}(t) := \|E(t) Z_0\|
+
\int_0^t \|E(t-s)G(s)\|\,ds.
\]
Hence the error is driven by the projection error on the full-order matrix exponential.

\subsection{Propagation regime}

On $(T_L,T)$, the equations are:
\[
\dot Z = -i\Omega Z,
\qquad
\dot{ Z}_r = -i\Omega_r  Z_r.
\]

Let $\tau = t-T_l$. Both full and reduced flows are unitary, implying for $t \geq T_l$:
\begin{align*}
\|Z(t)-\tilde Z_r(t)\| &= \|e^{-i\Omega \tau}Z(T_L)- U(t) Z_r(t)\| \\
& =  \| e^{-i\Omega \tau}Z(T_L) -  e^{-i\Omega \tau} U_r  e^{i\Omega_r \tau} e^{-i\Omega_r \tau} Z_r(T_L) \| \\
& = \| Z(T_L) - U_r Z_r(T_L) \| \\
& = \| Z(T_L) - U_r U_r^H A_l Z_l(T_L) \|\\
& = \| Z(T_L) - P_r \tilde Z_l(T_L) \|
\end{align*}
where we used the initial condition (\ref{eq:wave_initial_condition}). Next, we add and substract $\tilde{Z}_l(T_L) = P_l \tilde{Z}_l(T_L)$ in the relation above and use the triangle inequality. This yields:
\begin{align*}
\|Z(t)-\tilde Z_r(t)\| &\leq \|Z(T_L)-\tilde Z_l(T_L)\| +  \|(I-P_r) P_l \tilde Z_l(T_L) \| \\
& \leq \|Z(T_L)-\tilde Z_l(T_L)\| +  \|(I-P_r) P_l \| \| \tilde Z_l(T_L) \|
\end{align*} 

The operator norm $\|(I-P_r) P_l \| $ is related to the largest principal angle $\theta \in (0,\pi/2)$ between the ranges of $U_r$ and $A_l$ (see equation (13) in \cite{bjorck1973numerical}). It can be computed from the smallest singular value of $U_r^H A_l$ as described in Algorithm \ref{alg:offline}. This yields:
\begin{equation*}
\|Z(t)-\tilde Z_r(t)\| \leq \|Z(T_L)-\tilde Z_l(T_L)\| + \sin (\theta) \| \tilde Z_l(T_L) \|
\end{equation*} 

\subsection{Error estimate and gluing condition}
\label{subsec:gluing_condition}

Combining both regimes yields for $t \geq T_L$:
\begin{equation}
\|Z(t)-\tilde Z_r(t)\|
\le
\varepsilon_{\mathrm{loading}}(T_L)
+
\sin(\theta)\| \tilde Z_l(T_L)\|.
\label{eq:gluing_condition}
\end{equation}

The estimate shows that the long-time accuracy is governed  by the loading regime error and the geometric alignment between the reduced subspaces, which we refer as the gluing condition at the interface $T_L$. An important consequence is that the basis $U_r$ should not be identified independently from $A_l$. This motivates the simple procedure proposed in Algorithm \ref{alg:offline} which proved effective in our numerical experiments. 

\section{Lamb waves interacting with a damage}
\label{sec:numerical_application}

In this section, we apply our symplectic DLRA framework to a 2D plane strain elasticity problem featuring Lamb waves interacting with a damage. We investigate the influence of the rank $l$ of the static basis (used in the loading regime) and the rank $r \leq l$ of the dynamic basis (used in the wave propagation regime). Accuracy is assessed through error metrics such as full field error over time, QoI reconstruction and energy drift, while efficiency is evaluated via computational costs. A comparison with a baseline modal-based ROM using a static modal basis demonstrates the framework's advantages.

\newpage

\subsection{Numerical setup}
\label{subsec:numerical_setup}

\subsubsection{Model description}
\label{subsubsec:model_description}

\begin{figure}[t]
    \centering
    \includegraphics[width=1\textwidth]{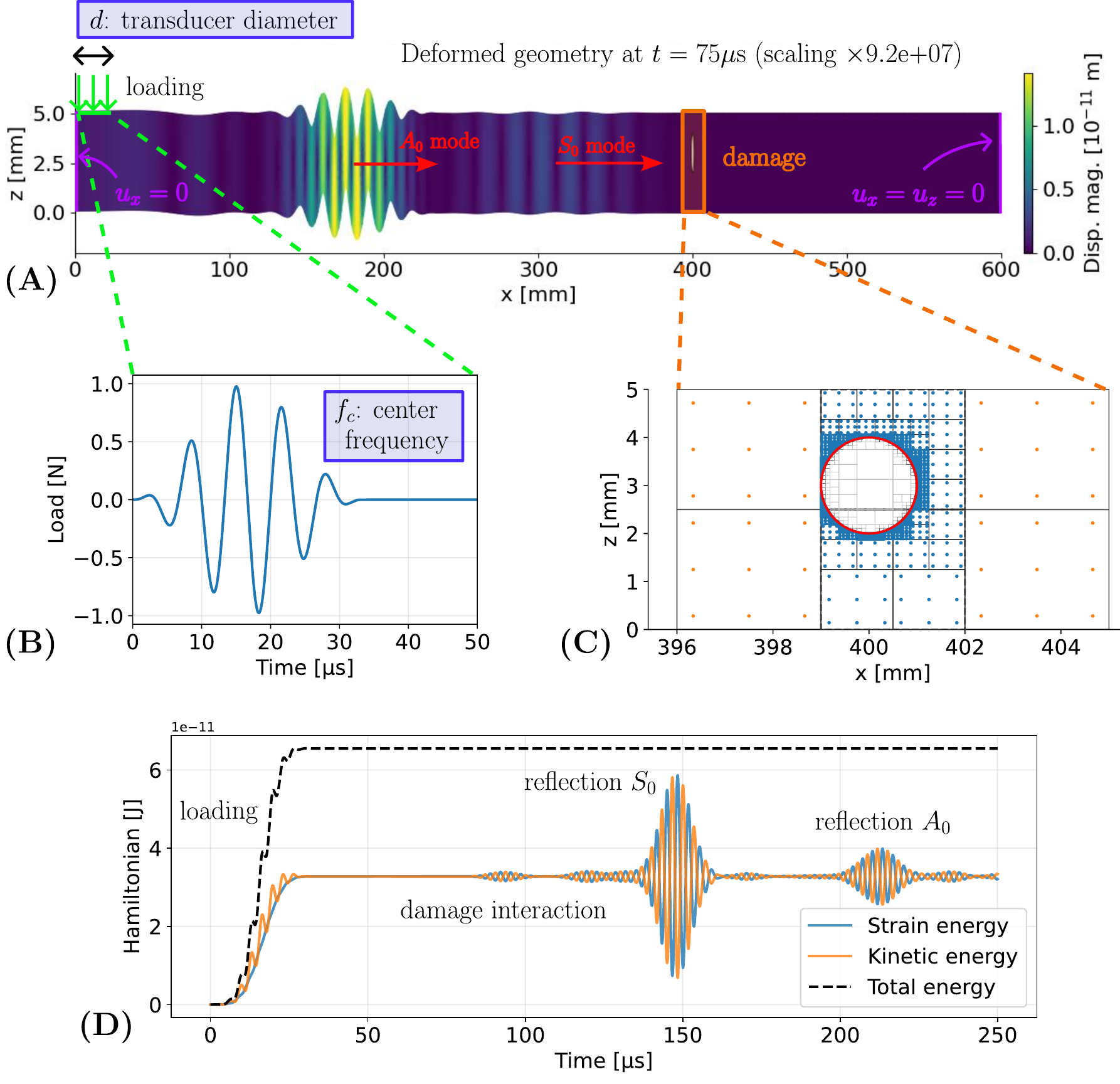}
    \caption{(A) Deformed geometry at $t=75\mu$s featuring the propagation of $A_0$ and $S_0$ Lamb wave modes. (B) Temporal profile of the applied loading. (C) Close-up of the spectral element mesh featuring cut cells around the hole. The dots represent the quadrature points. (D) Time evolution of the system's energy over time.}
    \label{fig:numerical_setup}
\end{figure}

The problem geometry is illustrated in Figure~\ref{fig:numerical_setup}-A. It consists of a rectangular aluminum plate with a piezo-transducer mounted on the top surface to excite ultrasonic Lamb waves \cite{duczek2014numerical}. A symmetry condition is enforced on the left edge ($u_x = 0$). The right edge is clamped to remove rigid body modes ($u_x = u_z = 0$). Initial conditions are null. The time-dependent amplitude of the load applied by the transducer is modeled as a sine-burst:
\begin{equation}
    F(t) = F_{\text{amp}} \sin( \omega t) \sin^2 \left(\frac{\omega t}{2n} \right) \quad \text{for} \quad t \leq \frac{n}{f_c},
    \label{eq:sine_burst}
\end{equation}
where $\omega = 2 \pi f_c$, with $f_c$ the central frequency and $n$ the number of cycles, see Figure~\ref{fig:numerical_setup}-B. A dummy total force of $1$N is applied over the diameter $d$ of the transducer. \\

The plate features a circular hole in its thickness representing an idealized damage in a SHM context. It is numerically discretized with the spectral cell method described in \cite{duczek2014numerical}. 
Specifically, the uncut cells use the standard SEM-GLL quadrature, while for cut cells a recursive quadtree
Gauss-Legendre quadrature resolves the damage geometry, see Figure \ref{fig:numerical_setup}-C. A row-sum lumped diagonal mass matrix is used to retain the explicit time-integration efficiency. No divergence behavior was observed in our numerical experiments, unlike the issues reported in \cite{duczek2014numerical}. An analysis of lumping strategies within the spectral cell method is beyond the scope of this paper. \\

Figure~\ref{fig:numerical_setup}-A shows the scaled deformed geometry at $t = 75\mu$s (before the wave-damage interaction) with the simulation parameters provided in Table \ref{tab:params} in the case $f_c = 150$kHz and $d = 3$mm. As reported in \cite{duczek2014numerical}, this configuration excites a mixture of $A_0$ and $S_0$ Lamb wave modes. Figure \ref{fig:numerical_setup}-D reports the time evolution of the system's Hamiltonian over time. The total energy increases during the loading regime, and is constant in the wave propagation regime. This confirms an accurate energy conservation at the time discrete level with the Stormer-Verlet time marching scheme.

\begin{table}
\centering
\begin{tabular}{ll}
\textbf{Parameter} & \textbf{Value} \\
\hline
Length \(L_x\) & 600 mm \\
Thickness \(L_z\) & 5 mm \\
Density \(\rho\) & 2700 kg/m\(^3\) \\
Young's modulus \(E\) & 70 GPa \\
Poisson's ratio \(\nu\) & 0.33 \\
Circular damage  & radius: 1mm, center: (400,3) mm \\
\hline 
Number of elements \(N_x\) & 200 \\
Number of elements \(N_z\) & 2 \\
Polynomial order \(p_x\) & 3 \\
Polynomial order \(p_z\) & 4 \\
Number of dofs \( n \) & 10,818 \\
\hline
Time step \(\Delta t\) & 45.2 ns ($\Delta t_{\textrm{CDM}} = 56.6$ ns)\\ 
Total duration \(T\) & 250 \(\mu\)s ($T_L = 50~\mu$s) \\
 Number of time steps \( N_t \) & 5526 (1105 + 4421) \\
 \hline
 Loading center frequency $f_c$ & [100, 200] kHz \\
Transducer diameter $d$ & [3, 12] mm \\
\hline
 Training points $\mathcal{D}^{\text{training}}$ & 100 uniform samples \\
 Test points $\mathcal{D}^{\text{query}}$ & 50 random samples ($m = 50$) \\
 \hline
\end{tabular}
\caption{Model parameters and design of experiments.}
\label{tab:params}
\end{table}

\subsubsection{Transducer design problem}
\label{subsubsec:transducer}

We focus on two design parameters of the transducer: the loading central frequency $f_c$ and the diameter $d$ of the spatial loading support. We consider the following ranges of interest:
\begin{equation*}
\mu = [f_c;d] \in [100,200]\text{kHz} \times [3,12]\text{mm}.
\end{equation*}
Accordingly, the loading regime duration for $f_c$ varying in [100, 200] kHz is defined as $T_L = n/ \min(f_c) = 50\mu\text{s}$. \\

We uniformly sample this two-dimensional parameter space with 10 points per dimension, leading to a discrete set $\mathcal{D}^{\text{training}}$ of 100 training parameters. We then randomly sample the parameter space with 50 points to build the discrete set $\mathcal{D}^{\text{query}}$. We mention here that a larger number of samples could not be used because of storage limitations on disk (1 full-order simulation $\sim$1Gb). 

\subsection{Dynamical low rank structure of the full-order model}

We first numerically demonstrate the low-rank structure of the studied problem through expensive computations. The objective is to show that the complexified FOM matrix solution $Z(t) \in \mathbb{C}^{n \times m}$, with $n = 10,818$ and $m=100$, admits a rank $r \ll \min(n,m) = 100$ for all times (again, we realize that a larger $m$ would have been more illustrative, but was impossible because of storage on disk limitations). To do so, we compute an instantaneous SVD of the complefixied FOM matrix solution for a subsampling of the available 5526 time steps. Then we compute the minimum rank $r(t)$ required to capture a prescribed threshold of the total matrix energy. Specifically, at time instance $t$ we solve:
\begin{equation}
\min_{r \in \{ 1, \cdots, 50\}} r(t) \quad \text{s.t.} \quad 1 - \frac{\sum_{i=1}^r \sigma_i^2(t)}{\sum_{i=1}^m \sigma_i^2(t)} \leq \delta
\end{equation}
for $\delta = 10^{-6}, \cdots 10^{-2}$. The results are displayed on Figure \ref{fig:SVD_ranks}-A. The rank grows (in average) during the loading regime and stabilizes after, as anticipated in section \ref{subsec:analytical_left_basis}. Notably, we remark that an instantaneous rank of 10 is enough to capture a fraction $1 - 10^{-3}$ of the total energy, suggesting a low-rank structure locally in time and globally in parameters.

\begin{figure}[t]
    \centering
    \includegraphics[width=1\textwidth]{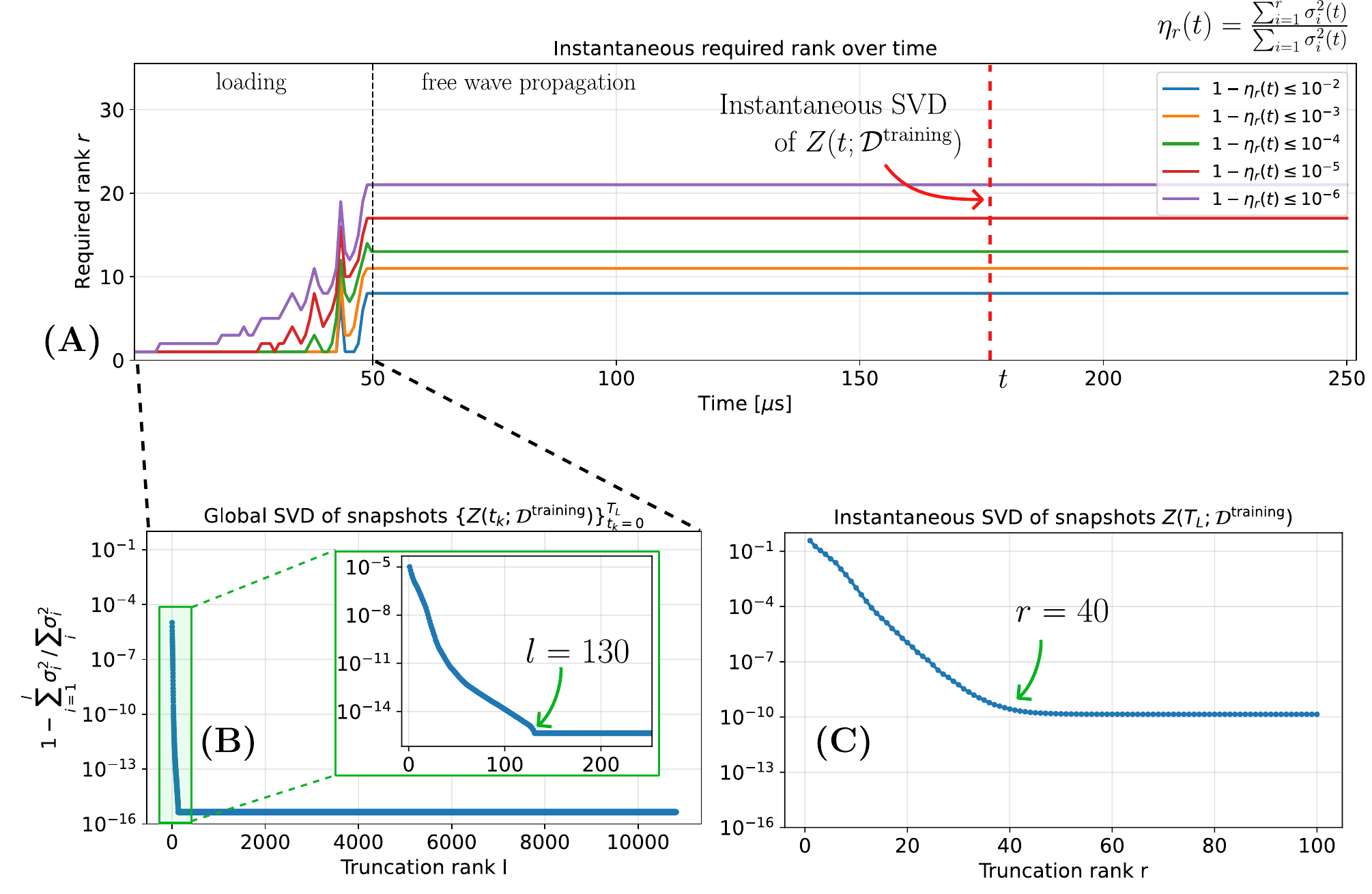}
    \caption{(A) Instantaneous truncation rank to capture a prescribed fraction of the complexified FOM matrix energy. (B) SVD of the concatenated complexified FOM matrices on the time interval $(0,T_L)$. (C) Instantaneous SVD of the complexified FOM matrix at time $T_L$.}
    \label{fig:SVD_ranks}
\end{figure}

\subsection{Rank selection and ROM construction}

The off-line phase uses the training data (100 full-order simulations) to prepare distinct reduced bases for the loading and the wave propagation regimes. The gluing condition described in section \ref{sec:error_analysis} must be respected at time $T_L$ to ensure consistency between the two bases. \\ 

\paragraph{Loading regime basis ($A_l$)} Following Algorithm \ref{alg:offline}, we compute an SVD of the concatenated complexified FOM matrix solutions on the time interval $(0,T_L)$. We use $33\%$ of the available 1105 time steps to reduce the computational effort, resulting in a concatenated matrix of dimension $10,818 \times 36,500$. Figure \ref{fig:SVD_ranks}-B displays the fraction of the unresolved energy depending on the truncation rank. The curve exhibits a corner at $l=130$, revealing that a larger truncation rank does not significantly change the approximation accuracy. This moderately large rank is used to define the complex orthonormal basis $A_l$ used in the loading regime ROM. \\

\paragraph{Initial left basis ($U_r$)} We then compute an instantaneous SVD of the $10,818 \times 100$ complexified FOM matrix solution at time $T_L$. Figure \ref{fig:SVD_ranks}-C displays the fraction of the unresolved energy depending on the truncation rank. The curve exhibits a corner at $r=40$, indicating that a greater rank does not significantly reduce the truncation error. However, the rank cannot be selected based on this information only. Indeed, the gluing condition discussed in section \ref{sec:error_analysis} must be respected. It constrains the span of $U_r$ to be closely aligned with that of $A_l$. Thus, we compute the largest principal angle between $\mathcal{U}_r$ and $\mathcal{A}_l$ as a function of the truncation rank $r$ (and $l$ to provide a more complete information), see Figure \ref{fig:angles}-A. Based on this plot, truncation ranks 10-20-30 yield close subspace alignements within less than 1deg. Arguably, $r=30$ should be selected to simultanously capture the largest fraction of the snapshot matrix energy and comply with the gluing condition. We proceed however with the a posteriori choice $r=20$, which leads to slightly better results, as suggested on Figure \ref{fig:angles}-B described below. 

\begin{figure}[t]
    \centering
    \includegraphics[width=1\textwidth]{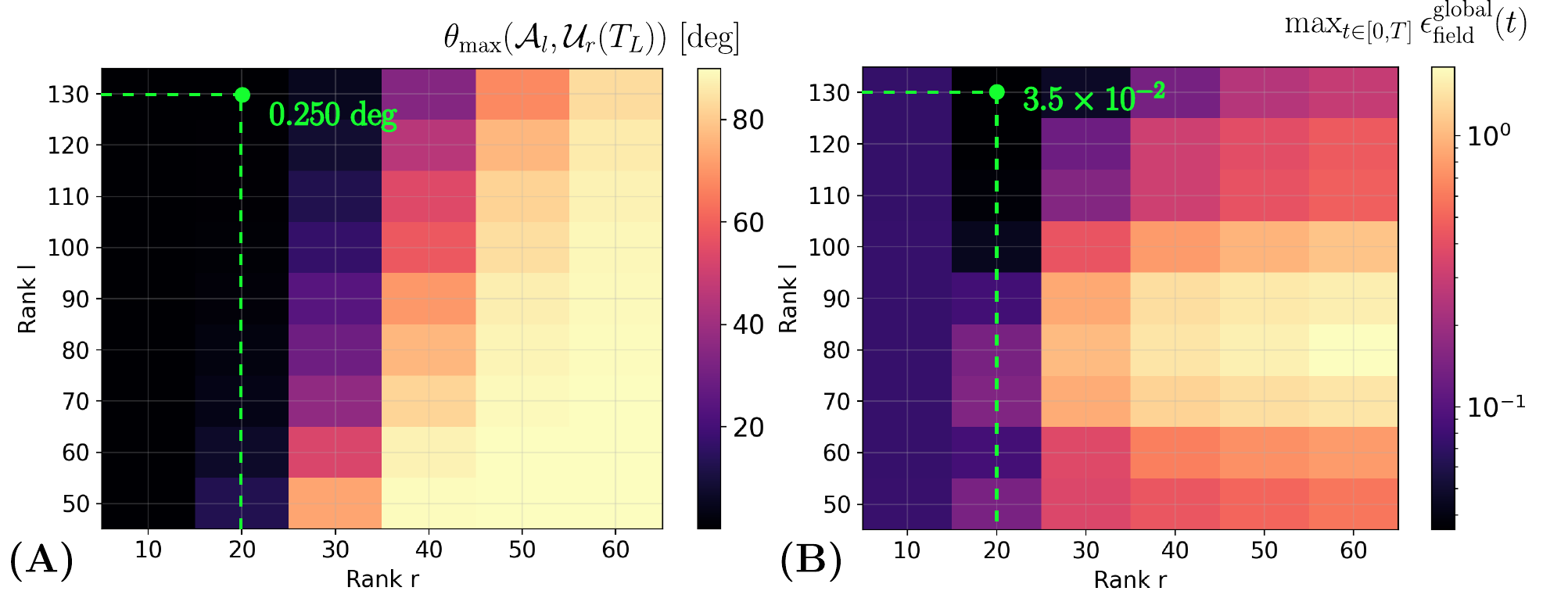}
    \caption{(A) Largest principal canonical angle between the subspaces spanned by the basis $A_l$ and $U_r$ as a function of the truncation ranks $r$ and $l$. (B) Maximum displacement field error over time as a function of the truncation ranks. The ROM uses $l=130$ and $r=20$. }
    \label{fig:angles}
\end{figure}

\subsection{Accuracy performances}

We now report the results obtained in the on-line phase, see Algorithm \ref{alg:online}, with the selected ranks $l=130$ and $r=20$. The following error metrics are used to compare the FOM and ROM solutions at the 50 test samples: \\
\begin{itemize}
\item Instantaneous displacement field reconstruction accross all query parameters:
\begin{equation*}
\epsilon_{\text{field}}^{\text{global}}(t) = \frac{\Vert U_{\text{FOM}}(t) - U_{\text{ROM}}(t) \Vert_F }{\Vert U_{\text{FOM}}(t) \Vert_F}
\end{equation*}
where $U(t) \in \mathbb{R}^{n \times m}$ is the displacement matrix solution at time $t$ for all parameters.
\item Instantaneous displacement field reconstruction at a single query parameter:
\begin{equation*}
\epsilon_{\text{field}}(t;\mu) = \frac{\Vert u_{\text{FOM}}(t;\mu) - u_{\text{ROM}}(t;\mu) \Vert }{\Vert u_{\text{FOM}}(t;\mu) \Vert}
\end{equation*}
where $u(t;\mu) \in \mathbb{R}^{n}$ is the displacement vector solution for the parameter $\mu$.
\item Time series data reconstruction at sensor location $P=(x_s,z_s)$ for a single query parameter:
\begin{equation*}
\epsilon_{\text{series}}(P,\mu) = \frac{\Vert \{ u_{\text{FOM}} (x_s,z_s,t_k;\mu) \}_{k=1}^K -  \{ u_{\text{ROM}} (x_s,z_s,t_k;\mu) \}_{k=1}^K \Vert}{\Vert \{ u_{\text{FOM}} (x_s,z_s,t_k;\mu) \}_{k=1}^K \Vert}
\end{equation*}
where $\{ u(x_s,z_s,t_k;\mu) \}_{k=1}^K \in \mathbb{R}^{K}$ is the time series displacement data at $P$ for the query parameter $\mu$.\\ 
\end{itemize}

\paragraph{Time marching} We use the same time marching scheme in the loading regime to compute the FOM and the ROM solutions with the procedure described in section \ref{subsec:ROM_loading_regime}. This removes additional discrepancies due to a mismatch between two different time integrators. \\

\paragraph{Displacement field reconstruction} Figure \ref{fig:heatmap_parameter_space}-B shows that the maximum of $\epsilon_{\text{field}}^{\text{global}}$ over time is $3.5\%$ for the selected ranks. It is reached at $t=160\mu$s, see blue curve in Figure \ref{fig:best_pair}, which essentially corresponds to the reflection of the $S_0$ wave at the right edge, where strong energy conversion occurs, see Figure \ref{fig:numerical_setup}-D. Figure \ref{fig:heatmap_parameter_space}-A then shows the displacement field accuracy per query parameters at time $t = 175\mu$s. The error is maximum ($7.49 \%$) at $f_c = 129$kHz and $d=5.5$mm, see Figure \ref{fig:heatmap_parameter_space}-B. It is minimum ($0.845 \%$) at $f_c = 165$kHz and $d=7.6$mm, see Figure \ref{fig:heatmap_parameter_space}-C. Although the deformed geometries are significantly affected by the transducer design parameters, the ROM demonstrate satisfactory accuracy (error $\sim 1 \%$) accross the whole parameter space. \\

\paragraph{Time series data reconstruction} We then consider three sensor locations with different positions of interest relative to the damage: $P_1 = (300,5)$mm (upstream), $P_2 = (500,5)$ mm (downstream) and $P_3 = (399,3)$mm (on boundary). Figure \ref{fig:sensors} shows the reconstruction over time of the $x$-component of the displacement field at the sensor locations for three randomly selected query parameters. The incident waves arrive in the expected order: $P_1$, $P_3$, then $P_2$. The error indicator $\epsilon_{\text{series}}$ is below $4 \%$ in all cases, and similarly when considering the $z$-component of the displacement field. 

\begin{figure}[H]
    \centering
    \includegraphics[width=0.9\textwidth]{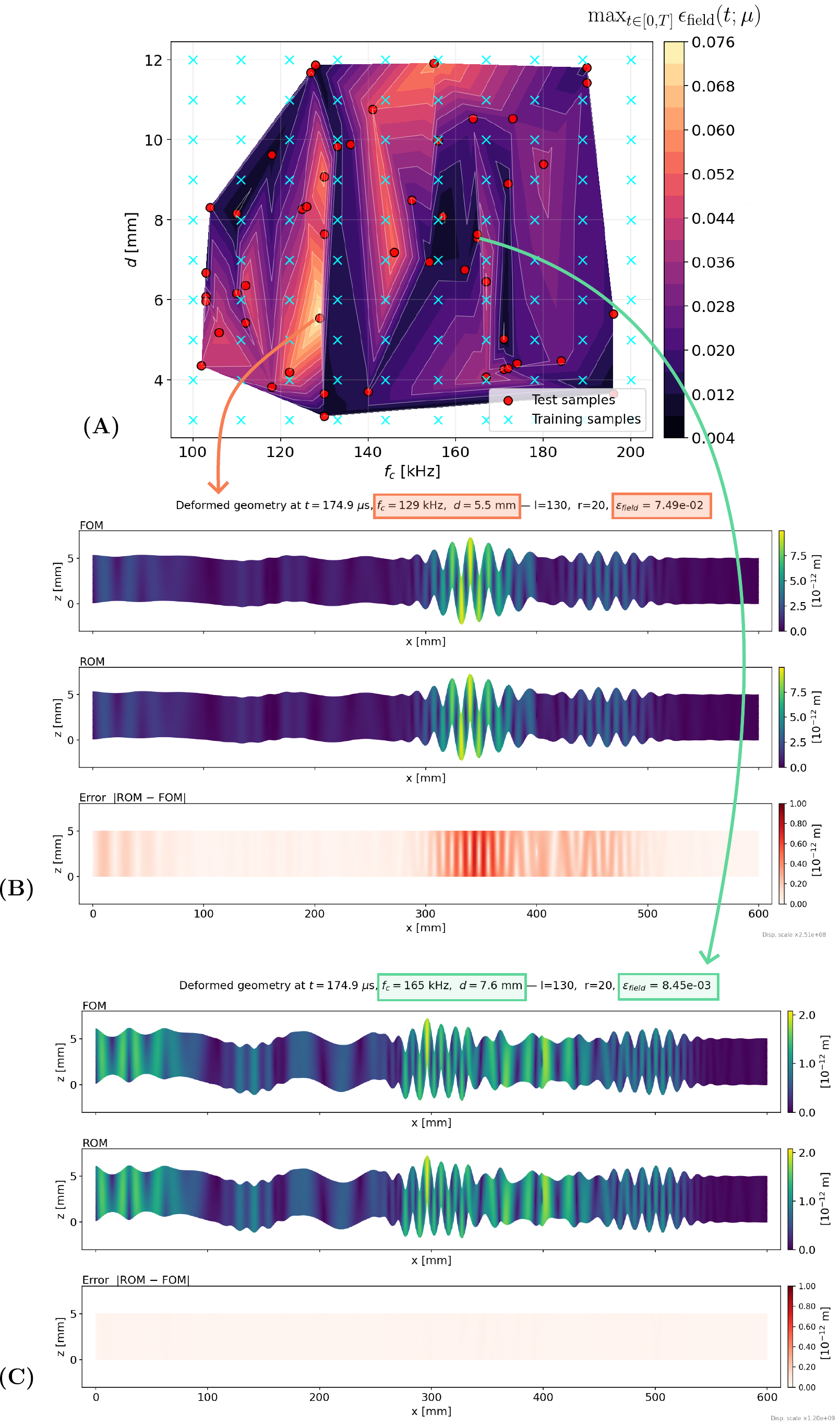}
    \caption{(A) Distribution of the displacement field error accross the parameter space. (B) ROM-FOM comparison at $f_c = 129$kHz and $d=5.5$mm. (C) Comparison at $f_c = 165$kHz and $d=7.6$mm.}
    \label{fig:heatmap_parameter_space}
\end{figure}

\begin{figure}[H]
    \centering
    \includegraphics[width=1\textwidth]{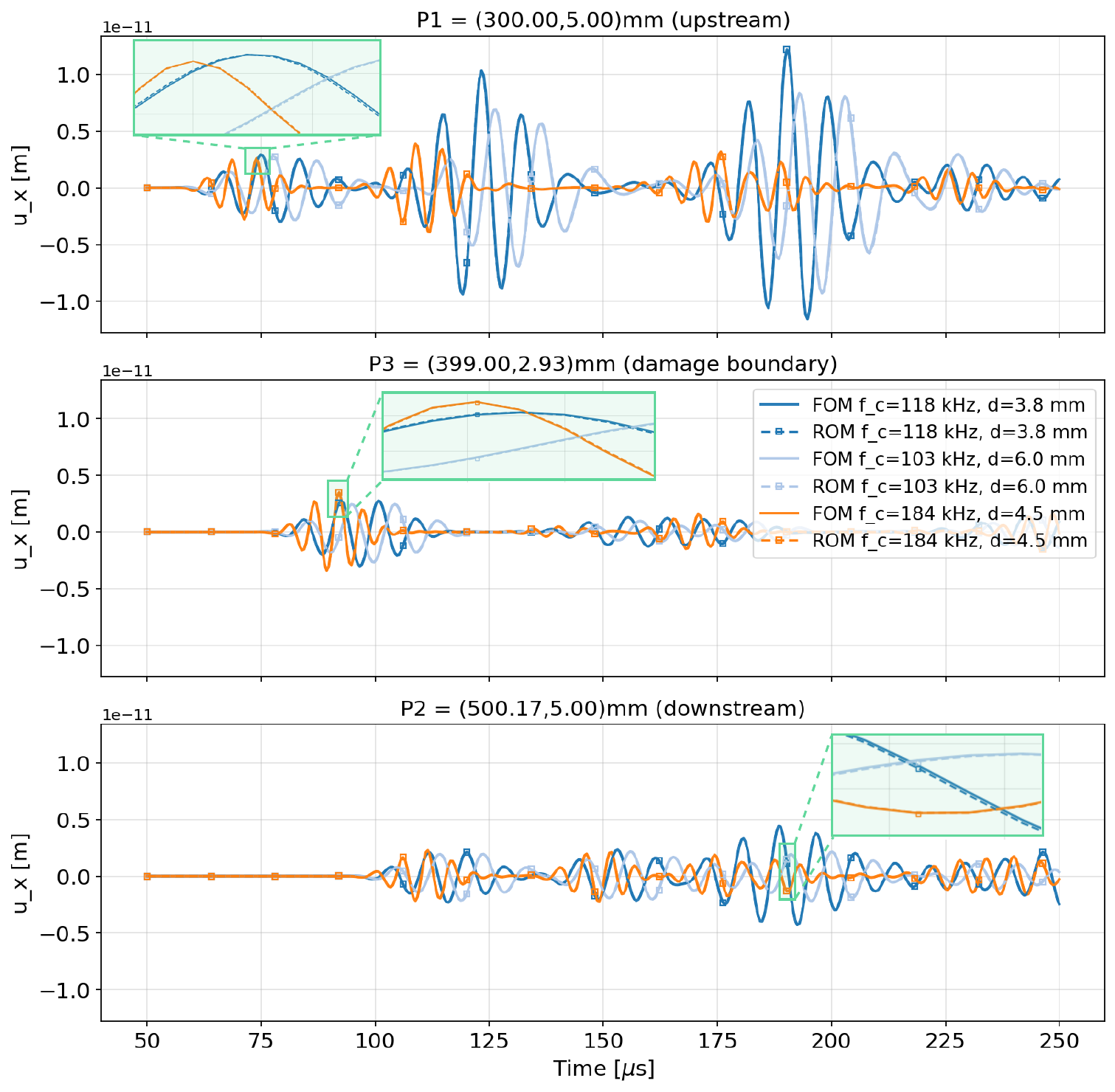}
    \caption{Time series displacement data reconstruction at sensor locations $P_1$-$P_2$-$P_3$ for three randomly selected query parameters. The error metric $\epsilon_{\text{series}}$ is below $4\%$ in all cases.}
    \label{fig:sensors}
\end{figure}

\subsection{Energy conservation and long-time stability}

We next assess the long-time stability of the ROM solution by monitoring the solution norm drift during the conservative regime ($t\geq T_L$):
\begin{equation}
\delta \mathcal{E}(t) = \frac{\vert \mathcal{E}(t) -  \mathcal{E}(T_L) \vert}{\vert \mathcal{E}(T_L) \vert}.
\end{equation}
where $\mathcal{E}(t) = \frac{1}{2}\Vert Z(t) \Vert_F^2$. The evolution of $\delta\mathcal{E}(t)$ is shown by the red curve in Figure~\ref{fig:best_pair}. The relative drift remains at the level of machine precision throughout the conservative regime, indicating that the reduced dynamics exhibit excellent long-time stability. This behavior follows directly from the use of the closed-form, structure-preserving propagator (\ref{eq:ROM_analytical}), which exactly preserves the unitary structure in the wave-propagation regime up to floating-point roundoff.

\begin{figure}[t]
    \centering
    \includegraphics[width=1\textwidth]{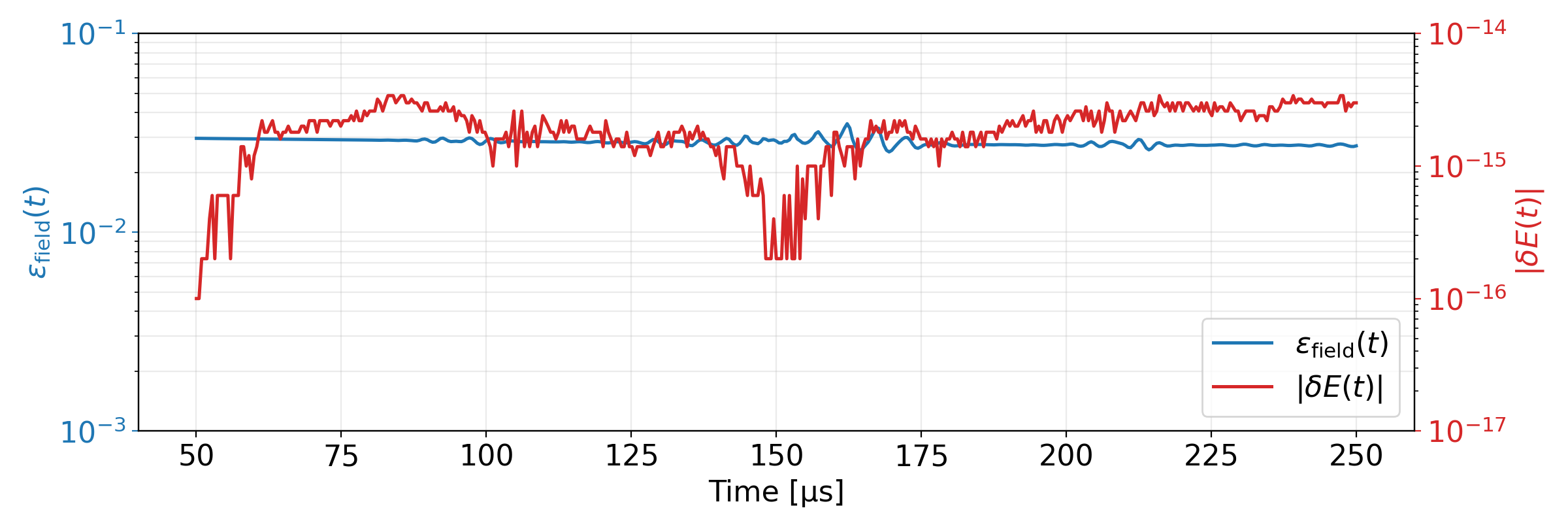}
    \caption{Blue curve: evolution of the displacement field reconstruction error over time. Red curve: ROM energy fluctuation over time.}
    \label{fig:best_pair}
\end{figure}

\subsection{Computational performance}

All the computations were performed using Python 3.14 on a single 13th Gen Intel(R) Core(TM) i7-13800H CPU (2.5~GHz) with 64~GB RAM. Wall-clock times are averaged over 20 runs. Performance is evaluated over the complete set of 50 query parameters for two representative quantities of interest (QoIs): time series displacement data reconstruction at sensor locations $P_1$-$P_2$-$P_3$, and full displacement field reconstruction at the final time $T=250\mu$s. \\ 

\paragraph{Baseline modal ROM} The DLRA ROM is compared against a classical modal ROM. A posteriori experiments indicate that approximately 500 vibration modes are required to achieve an accuracy comparable to that of the proposed method across the parameter set. The modal ROM is also divided into the loading and the wave propagation regimes. During the loading, the reduced dynamics are obtained from (\ref{eq:ROM_loading_online}) by replacing the basis $A_l$ with the truncated modal basis $\Psi(:,1:500)$, where $\Psi$ denotes the full modal basis defined in (\ref{eq:computation_omega}). During the wave propagation, the analytical propagator (\ref{eq:wave_exact_solution}) is employed using truncated matrices in place of the full spectral decomposition $\Psi$ and $\Lambda$.\\

\paragraph{QoI 1: Sensor time-history reconstruction} During the loading regime, the DLRA ROM achieves a speedup of $250\times$ relative to the FOM, demonstrating that the rank $l=130$ remains computationally efficient. This is approximately $4\times$ faster than the modal ROM owing to its substantially smaller reduced basis. In the wave-propagation regime, the closed-form approximations used in the ROMs yield speedups of $1000\times$ for the modal ROM and $3000\times$ for the DLRA ROM. Furthermore, unlike the FOM which computes the complete displacement field before extracting sensor responses, both ROMs reconstruct the sensor signals directly using masked reduced bases. This provides an additional acceleration of approximately $5\times$ for the modal ROM and $40\times$ for the DLRA-ROM. Overall, the DLRA ROM achieves a speedup of approximately $820\times$, corresponding to a fourfold improvement over the modal ROM.
\\

\paragraph{QoI 2: Displacement field reconstruction} Similarly, we consider the reconstruction of the complete displacement field at the final time $T=250\mu$s. Since the analytical propagator is evaluated only once at the requested time, both reduced models provide substantial computational savings. The modal ROM achieves an overall speedup of approximately $280\times$, whereas the proposed DLRA-ROM reaches nearly $1000\times$. This additional gain is primarily due to the significantly smaller reduced dimension ($r=20$), which lowers both the propagation and reconstruction costs. 

\begin{table}
\centering
\begin{tabular}{lccc}
 & FOM & Modal ROM ($r=500$) & DLRA ROM ($r=20$) \\
 \hline
 RHS & 0.0038 & 0.0010 & 0.0024 \\
$(0,T_L)$ & 10 & 0.15 & 0.041 \\
$(T_L,T)$ & 40 & 0.040 & 0.013 \\
QoI  & 0.19 & 0.039 & 0.0044 \\
\textbf{Total}  & \textbf{50} & \textbf{0.23 (220$\times$)} & \textbf{0.061 (820$\times$)} \\
\end{tabular}
\caption{Wall-clock time in seconds averaged over 20 runs for computing the sensor data reconstruction accross all the 50 query parameters. RHS: construction of the parameter-dependent operators.}
\end{table}

\section{Conclusions}
\label{sec:conclusions}

This work introduced a structure-preserving ROM framework for parametric simulations of elastic guided waves. The proposed off-line/on-line strategy exploits the distinct characteristics of the loading and wave propagation regimes. During the loading, the increasing complexity of the solution is efficiently captured by a symplectic Galerkin projection, whereas the constant rank after the excitation motivates the use of the DLRA method. In the simplified parametric setting considered in this work, we derived a closed-form solution of the nonlinear left basis evolution equation, leading to an analytical, structure-preserving reduced propagator during the conservative regime. Consequently, no time integration is required after the loading phase, substantially reducing the on-line computational cost. The proposed approach was validated on a two-dimensional elastodynamic problem involving dispersive guided waves interacting with a damage. The numerical experiments demonstrated that the method combines high compression ratios, accurate QoI reconstructions, long-time energy conservation, and computational speedups reaching three orders of magnitude compared with the full-order model. These results illustrate the potential of dynamical low-rank techniques for guided wave propagation problems arising in many-query SHM applications. The main limitation of the present framework lies in the simplified parametric dependence, where the parameters affect only the loading. Extending the proposed methodology to more general parameterized operators, including variations in material properties or damage configurations, constitutes a natural direction for future work. 

\bibliographystyle{siamplain}
\bibliography{references}

@article{Benner_SIAM_2015,
author = {Benner, Peter and Gugercin, Serkan and Willcox, Karen},
title = {A Survey of Projection-Based Model Reduction Methods for Parametric Dynamical Systems},
journal = {SIAM Review},
volume = {57},
number = {4},
pages = {483-531},
year = {2015}
}

@book{hesthaven2016certified,
  title={Certified reduced basis methods for parametrized partial differential equations},
  author={Hesthaven, Jan S and Rozza, Gianluigi and Stamm, Benjamin and others},
  volume={590},
  year={2016},
  publisher={Springer}
}

@article{chen2017reduced,
  title={Reduced basis methods for uncertainty quantification},
  author={Chen, Peng and Quarteroni, Alfio and Rozza, Gianluigi},
  journal={SIAM/ASA Journal on Uncertainty Quantification},
  volume={5},
  number={1},
  pages={813--869},
  year={2017},
  publisher={SIAM}
}

@article{Goutaudier_GTP_2024,
    author = {Goutaudier, Dimitri and Schiffmann, Jürg and Nobile, Fabio},
    title = {Parametric Reduced Order Model of a Gas Bearings Supported Rotor},
    journal = {Journal of Engineering for Gas Turbines and Power},
    volume = {146},
    number = {1},
    pages = {011002},
    year = {2023},
    month = {10},
    issn = {0742-4795}
}

@article{peherstorfer2020model,
  title={Model reduction for transport-dominated problems via online adaptive bases and adaptive sampling},
  author={Peherstorfer, Benjamin},
  journal={SIAM Journal on Scientific Computing},
  volume={42},
  number={5},
  pages={A2803--A2836},
  year={2020},
  publisher={SIAM}
}

@article{barnett2022quadratic,
  title={Quadratic approximation manifold for mitigating the Kolmogorov barrier in nonlinear projection-based model order reduction},
  author={Barnett, Joshua and Farhat, Charbel},
  journal={Journal of Computational Physics},
  volume={464},
  pages={111348},
  year={2022},
  publisher={Elsevier}
}

@article{reiss2018shifted,
  title={The shifted proper orthogonal decomposition: A mode decomposition for multiple transport phenomena},
  author={Reiss, Julius and Schulze, Philipp and Sesterhenn, J{\"o}rn and Mehrmann, Volker},
  journal={SIAM Journal on Scientific Computing},
  volume={40},
  number={3},
  pages={A1322--A1344},
  year={2018},
  publisher={SIAM}
}

@article{goutaudier2021proper,
  title={{Proper Generalized Decomposition with time adaptive space separation for transient wave propagation problems in separable domains}},
  author={Goutaudier, Dimitri and Berthe, Laurent and Chinesta, Francisco},
  journal={Computer Methods in Applied Mechanics and Engineering},
  volume={380},
  pages={113755},
  year={2021},
  publisher={Elsevier}
}

@article{goutaudier2022exploring,
  title={{Exploring space separation techniques for 3D elastic waves simulations}},
  author={Goutaudier, Dimitri and Berthe, Laurent and Chinesta, Francisco},
  journal={Computational Mechanics},
  volume={69},
  number={5},
  pages={1147--1163},
  year={2022},
  publisher={Springer}
}

@article{nicard2025situ,
  title={In-situ monitoring of $\mu$m-sized electrochemically generated corrosion pits using Lamb waves managed by a sparse array of piezoelectric transducers},
  author={Nicard, Cyril and R{\'e}billat, Marc and Devos, Olivier and El May, Mohamed and Letellier, Frederic and Dubent, S{\'e}bastien and Thomachot, M and Fournier, M and Masse, P and Mechbal, Nazih},
  journal={Ultrasonics},
  volume={147},
  pages={107527},
  year={2025},
  publisher={Elsevier}
}

@article{fendzi2016general,
  title={A general Bayesian framework for ellipse-based and hyperbola-based damage localization in anisotropic composite plates},
  author={Fendzi, Claude and Mechbal, Nazih and Rebillat, Marc and Guskov, Mikhail and Coffignal, G},
  journal={Journal of Intelligent Material Systems and Structures},
  volume={27},
  number={3},
  pages={350--374},
  year={2016},
  publisher={SAGE Publications Sage UK: London, England}
}

@article{taddei2021space,
  title={Space-time registration-based model reduction of parameterized one-dimensional hyperbolic PDEs},
  author={Taddei, Tommaso and Zhang, Lei},
  journal={ESAIM: Mathematical Modelling and Numerical Analysis},
  volume={55},
  number={1},
  pages={99--130},
  year={2021},
  publisher={EDP Sciences}
}

@article{hesthaven2022rank,
  title={Rank-adaptive structure-preserving model order reduction of Hamiltonian systems},
  author={Hesthaven, Jan S and Pagliantini, Cecilia and Ripamonti, Nicol{\`o}},
  journal={ESAIM: Mathematical Modelling and Numerical Analysis},
  volume={56},
  number={2},
  pages={617--650},
  year={2022},
  publisher={EDP Sciences}
}

@article{pagliantini2021dynamical,
  title={Dynamical reduced basis methods for Hamiltonian systems},
  author={Pagliantini, Cecilia},
  journal={Numerische Mathematik},
  volume={148},
  number={2},
  pages={409--448},
  year={2021},
  publisher={Springer}
}

@article{musharbash2020symplectic,
  title={Symplectic dynamical low rank approximation of wave equations with random parameters},
  author={Musharbash, Eleonora and Nobile, Fabio and Vidli{\v{c}}kov{\'a}, Eva},
  journal={BIT Numerical Mathematics},
  volume={60},
  number={4},
  pages={1153--1201},
  year={2020},
  publisher={Springer}
}

@article{peng2016symplectic,
  title={Symplectic model reduction of Hamiltonian systems},
  author={Peng, Liqian and Mohseni, Kamran},
  journal={SIAM Journal on Scientific Computing},
  volume={38},
  number={1},
  pages={A1--A27},
  year={2016},
  publisher={SIAM}
}

@article{buchfink2019symplectic,
  title={Symplectic model order reduction with non-orthonormal bases},
  author={Buchfink, Patrick and Bhatt, Ashish and Haasdonk, Bernard},
  journal={Mathematical and Computational Applications},
  volume={24},
  number={2},
  pages={43},
  year={2019},
  publisher={MDPI}
}

@article{peng2022data,
  title={Data-driven model order reduction with proper symplectic decomposition for flexible multibody system},
  author={Peng, Haijun and Song, Ningning and Kan, Ziyun},
  journal={Nonlinear Dynamics},
  volume={107},
  number={1},
  pages={173--203},
  year={2022},
  publisher={Springer}
}

@article{koch2007dynamical,
  title={Dynamical low-rank approximation},
  author={Koch, Othmar and Lubich, Christian},
  journal={SIAM Journal on Matrix Analysis and Applications},
  volume={29},
  number={2},
  pages={434--454},
  year={2007},
  publisher={SIAM}
}

@article{feppon2018geometric,
  title={A geometric approach to dynamical model order reduction},
  author={Feppon, Florian and Lermusiaux, Pierre FJ},
  journal={SIAM Journal on Matrix Analysis and Applications},
  volume={39},
  number={1},
  pages={510--538},
  year={2018},
  publisher={SIAM}
}

@article{lubich2014projector,
  title={A projector-splitting integrator for dynamical low-rank approximation},
  author={Lubich, Christian and Oseledets, Ivan V},
  journal={BIT Numerical Mathematics},
  volume={54},
  number={1},
  pages={171--188},
  year={2014},
  publisher={Springer}
}

@article{ceruti2022unconventional,
  title={An unconventional robust integrator for dynamical low-rank approximation},
  author={Ceruti, Gianluca and Lubich, Christian},
  journal={BIT Numerical Mathematics},
  volume={62},
  number={1},
  pages={23--44},
  year={2022},
  publisher={Springer}
}

@article{nobile2026high,
  title={High-Order BUG Dynamical Low-Rank Integrators Based on Explicit Runge--Kutta Methods},
  author={Nobile, Fabio and Riffaud, S{\'e}bastien},
  journal={Journal of Scientific Computing},
  volume={107},
  number={3},
  pages={102},
  year={2026},
  publisher={Springer}
}

@article{ceruti2022rank,
  title={A rank-adaptive robust integrator for dynamical low-rank approximation},
  author={Ceruti, Gianluca and Kusch, Jonas and Lubich, Christian},
  journal={BIT Numerical Mathematics},
  volume={62},
  number={4},
  pages={1149--1174},
  year={2022},
  publisher={Springer}
}

@article{hairer2006geometric,
  title={Geometric numerical integration},
  author={Hairer, Ernst and Hochbruck, Marlis and Iserles, Arieh and Lubich, Christian},
  journal={Oberwolfach Reports},
  volume={3},
  number={1},
  pages={805--882},
  year={2006}
}

@book{bathe2006finite,
  title={Finite element procedures},
  author={Bathe, Klaus-J{\"u}rgen},
  year={2006},
  publisher={Klaus-Jurgen Bathe}
}

@book{su2009identification,
  title={Identification of damage using Lamb waves: from fundamentals to applications},
  author={Su, Zhongqing and Ye, Lin},
  year={2009},
  publisher={Springer Science \& Business Media}
}

@article{komatitsch1999spectral,
  title={The spectral element method for elastic wave equations—application to 2-D and 3-D seismic problems},
  author={Komatitsch, Dimitri and Vilotte, Jean-Pierre and Vai, Rossana and Castillo-Covarrubias, Jos{\'e} M and S{\'a}nchez-Sesma, Francisco J},
  journal={International Journal for numerical methods in engineering},
  volume={45},
  number={9},
  pages={1139--1164},
  year={1999},
  publisher={Wiley Online Library}
}

@article{duczek2014numerical,
  title={Numerical analysis of Lamb waves using the finite and spectral cell methods},
  author={Duczek, Sascha and Joulaian, M and D{\"u}ster, A and Gabbert, U},
  journal={International Journal for Numerical Methods in Engineering},
  volume={99},
  number={1},
  pages={26--53},
  year={2014},
  publisher={Wiley Online Library}
}

@article{bjorck1973numerical,
  title={Numerical methods for computing angles between linear subspaces},
  author={Bj{\"o}rck, Ake and Golub, Gene H},
  journal={Mathematics of computation},
  volume={27},
  number={123},
  pages={579--594},
  year={1973}
}
\end{document}